\RequirePackage[l2tabu, orthodox]{nag}

\PassOptionsToPackage{hidelinks, breaklinks}{hyperref}

\documentclass[submission,copyright,creativecommons,noncommercial,sharealike]{eptcs}
\providecommand{\event}{AiML 2026}

\usepackage{xparse}

\usepackage{mleftright}
\usepackage{mathtools}
\mathtoolsset{mathic}
\usepackage{amsthm}

    \usepackage[T1]{fontenc}
\DeclareFontShape{\encodingdefault}{ptm}{m}{scit}{<->ssub * ntxosf/m/scit}{}
\usepackage{amssymb}
    \let\symbb\mathbb
    \let\symcal\mathcal
    
    \let\symbf\mathbf
    \let\symfrak\mathfrak
\DeclareSymbolFont{stmry}{U}{stmry}{m}{n}
    \DeclareMathDelimiter\lBrack{\mathopen}{stmry}{"4A}{stmry}{"71} \DeclareMathDelimiter\rBrack{\mathclose}{stmry}{"4B}{stmry}{"79} \DeclareSymbolFont{wasy}{U}{wasy}{m}{n}
    \DeclareMathSymbol\mdwhtsquare{\mathord}{wasy}{"32} \newcommand{\mdlgwhtsquare}{\mdwhtsquare} \DeclareMathSymbol\mdwhtdiamondold{\mathord}{wasy}{"33}
    \newcommand{\mdwhtdiamond}{\mathchoice{\raisebox{-.09em}{$\displaystyle\mdwhtdiamondold$}}
             {\raisebox{-.09em}{$\mdwhtdiamondold$}}
             {\raisebox{-0.05em}{$\scriptstyle\mdwhtdiamondold$}}
             {\raisebox{-0.02em}{$\scriptscriptstyle\mdwhtdiamondold$}}} \renewcommand{\boxdot}{\mathord{\mdwhtsquare\kern-1.14ex\cdot\kern.55ex}} 

    \newcommand{\mathslash}{/} 

\newcommand{\smblkcircle}{\bullet} 

\usepackage{microtype}

\usepackage[british]{babel}
\usepackage[autostyle]{csquotes}

\usepackage{xcolor}

\usepackage{semantex}
\usepackage{enumitem}

\usepackage{underscore}
\usepackage{fnpct}

\usepackage{zref}
\usepackage{zref-clever}

\NewDocumentCommand{\zsubref}{m o m}{\IfValueTF{#2}{\zcref{#1}[#2]}{\zcref{#1}}~\zref{#3}}
\NewDocumentCommand{\citepage}{m}{p.~#1}

    \usepackage{scalerel} 

\newcommand{\conda}{(C1)}
\newcommand{\condb}{(C2)}
\newcommand{\fmp}{\textsc{fmp}} \newcommand{\tdefiff}{if} 

\setlist[description]{labelindent=0.5\leftmargin}

\newcommand*{\todo}[1]{\relax\ifmmode\textsf{todo: #1}\else\textsf{Todo: #1.}\fi}

\zcRefTypeSetup{equation}{abbrev}
\zcLanguageSetup{english}{
  cap = true ,
  type = page ,
    cap = false ,
  type = equation ,
    cap = false ,
}
\zcLanguageSetup{dutch}{
  cap = true ,
  type = page ,
    cap = false ,
  type = equation ,
    cap = false ,
}
\zcLanguageSetup{english}{
  type = commentary ,
    Name-sg = Commentary ,
    name-sg = commentary ,
    Name-pl = Commentaries ,
    name-pl = commentaries ,
  type = claim ,
    Name-sg = Claim ,
    name-sg = claim ,
    Name-pl = Claims ,
    name-pl = claims ,
  type = axiom ,
    Name-sg = Axiom ,
    name-sg = axiom ,
    Name-pl = Axiomata ,
    name-pl = axiomata ,
  type = notation ,
    Name-sg = Notation ,
    name-sg = notation ,
    Name-pl = Notations ,
    name-pl = notations ,
  type = convention ,
    Name-sg = Convention ,
    name-sg = convention ,
    Name-pl = Conventions ,
    name-pl = conventions ,
  type = terminology,
    Name-sg = Terminology ,
    name-sg = terminology ,
    Name-pl = Terminologies ,
    name-pl = terminologies ,
  type = problem ,
    Name-sg = Problem ,
    name-sg = problem ,
    Name-pl = Problems ,
    name-sg = problems ,
  type = construction ,
    Name-sg = Construction ,
    name-sg = construction ,
    Name-pl = Constructions ,
    name-pl = constructions ,
  type = lemma ,
Name-pl = Lemmata ,
    name-pl = lemmata ,
}

\newtheorem{maintheorem}{Theorem}

\zcsetup{countertype={maintheorem=theorem}}
\newtheorem{maintheoremagain}{Theorem}

\zcsetup{countertype={maintheoremagain=theorem}}

\newtheorem{theorem}{Theorem}[section]
\newtheorem{proposition}[theorem]{Proposition}
\newtheorem{lemma}[theorem]{Lemma}
\newtheorem{corollary}[theorem]{Corollary}

\theoremstyle{definition}
\newtheorem{definition}[theorem]{Definition}

\newtheorem{construction}[theorem]{Construction}

\theoremstyle{remark}

\AddToHook{env/definition/begin}{\zcsetup{countertype={theorem=definition}}}
\AddToHook{env/construction/begin}{\zcsetup{countertype={theorem=construction}}}
\AddToHook{env/example/begin}{\zcsetup{countertype={theorem=example}}}
\AddToHook{env/remark/begin}{\zcsetup{countertype={theorem=remark}}}
\AddToHook{env/lemma/begin}{\zcsetup{countertype={theorem=lemma}}}
\AddToHook{env/proposition/begin}{\zcsetup{countertype={theorem=proposition}}}
\AddToHook{env/corollary/begin}{\zcsetup{countertype={theorem=corollary}}}
\AddToHook{env/notation/begin}{\zcsetup{countertype={theorem=notation}}}

\NewTranslationFallback{Case}{Case}
\NewTranslation{dutch}{Case}{Geval}
\NewTranslation{english}{Case}{Case}
\NewTranslation{german}{Case}{Fall}

\newlist{proofcases}{enumerate}{2}
\setlist[proofcases,1]{wide=0pt, label={\protect\GetTranslation{Case}~\arabic*.}, font=\emph}
\setlist[proofcases,2]{wide=0.5em, leftmargin=0.5em, label={\protect\GetTranslation{Case}~\arabic{proofcasesi}\alph{proofcasesii}.}, font=\emph}

\zcsetup{
    countertype = {
        proofcasesi = case ,
        proofcasesii = case ,
    } ,
    counterresetby = {
        proofcasesii = proofcasesi ,
    }
}

\newlist{thmenumerate}{enumerate}{2}
\setlist[thmenumerate,1]{label=\normalfont(\roman*)}
\setlist[thmenumerate,2]{label=\normalfont(\alph*)}

\zcsetup{
    countertype = {
        thmenumeratei = item ,
        thmenumerateii = item ,
    } ,
    counterresetby = {
        thmenumerateii = thmenumeratei ,
    }
}

\NewDocumentEnvironment{subproof}{o}{\IfValueTF{#1}{\begin{proof}[#1]}{\begin{proof}}}{\end{proof}}

\NewDocumentEnvironment{appendixproof}{m}{\begin{proof}[\normalfont\bfseries Proof of \zcref{#1}]
}{\end{proof}}

\let\lleft\mleft
\let\rright\mright
\let\mmiddle\middle

\NewDocumentCommand{\miff}{}{\quad \text{iff} \quad}

\NewDocumentCommand\limplies{}{\DOTSB\rightarrow}

\NewDocumentCommand\ldiamond{}{\mdwhtdiamond}
\NewDocumentCommand\lbox{}{\mdwhtsquare}
\NewDocumentCommand\lcircle{}{\mdwhtcircle}

\NewDocumentCommand\suchthat{}{\;\,}
\NewDocumentCommand\suchthatshort{}{\;}

\ExplSyntaxOn
\keys_define:nn { my/math/quantifier }
    {
        quantifier .tl_set:N = \l__my_math_quantifier_quant_tl,
        quantifier .value_required:n = true,
        exists .meta:n = {quantifier = {\exists}},
        forall .meta:n = {quantifier = {\forall}},
        doparen .bool_set:N = \l__my_math_quantifier_doparen_bool,
        doparen .default:n = true,
        nextquant .bool_set:N = \l__my_math_quantifier_nextquant_bool,
        nextquant .default:n = true,
    }

\keys_precompile:nnN { my/math/quantifier } { quantifier = {}, doparen = false, nextquant = false }
\l__my_math_quantifier_reset_defaults_tl

\NewDocumentCommand{\MakeQuantifier}{m m m}
    {
        \group_begin:

        \tl_use:N \l__my_math_quantifier_reset_defaults_tl

        \keys_set:nn { my/math/quantifier } { #1 }

        \bool_if:NTF \l__my_math_quantifier_nextquant_bool
        { \tl_use:N \l__my_math_quantifier_quant_tl #2 \suchthatshort #3 }
        {
            \bool_if:NTF \l__my_math_quantifier_doparen_bool
            { \tl_use:N \l__my_math_quantifier_quant_tl #2 \suchthatshort\left\lbrack #3 \right\rbrack }
            { \tl_use:N \l__my_math_quantifier_quant_tl #2 . \suchthat #3 }
        }

        \group_end:
    }

\NewDocumentCommand{\ForAll}{o m m}
    {
        \IfValueTF { #1 }
        { \MakeQuantifier{ #1, forall }{ #2 }{ #3 } }
        { \MakeQuantifier{ forall }{ #2 }{ #3 } }
    }

\NewDocumentCommand{\Exists}{o m m}
    {
        \IfValueTF { #1 }
        { \MakeQuantifier{ #1, exists }{ #2 }{ #3 } }
        { \MakeQuantifier{ exists }{ #2 }{ #3 } }
    }
\ExplSyntaxOff

\NewDocumentCommand{\Set}{m o}{\IfValueTF{#2}{\lleft\{#1\;\mmiddle|\;#2\rright\}}{\lleft\{#1\rright\}}}

\NewDocumentCommand{\Function}{m m m}{#1\colon #2 \to #3}

\NewVariableClass\mClassDelim

\NewObject\mClassDelim\paren[left par=\lparen, right par=\rparen]

\NewObject\mClassDelim\tuple[left par=\langle, right par=\rangle]
\NewObject\mClassDelim\homotuple[left par=\langle, right par=\rangle]
\NewObject\mClassDelim\hetrotuple[left par=\langle, right par=\rangle]
\NewObject\mClassDelim\structuple[left par=\langle\nobreak, right par=\rangle]

\NewObject\mClassDelim\abs[left par=\lvert, right par=\rvert]
\NewObject\mClassDelim\norm[left par=\lVert , right par=\rVert]
\NewObject\mClassDelim\card[left par=\lvert, right par=\rvert]
\NewObject\mClassDelim\sem[left par=\lBrack, right par=\rBrack]

\NewVariableClass\mClassMath

\NewDocumentCommand\mathstylealg{m}{\symbf{#1}}
\NewDocumentCommand\mathstylea{m}{#1}
\NewDocumentCommand\mathstylecd{m}{#1}
\NewDocumentCommand\mathstylef{m}{#1}
\NewDocumentCommand\mathstylefr{m}{\symfrak{#1}}
\NewDocumentCommand\mathstylefrm{m}{\symfrak{#1}}
\NewDocumentCommand\mathstylefrl{m}{\symfrak{#1}}
\NewDocumentCommand\mathstylelog{m}{#1}
\NewDocumentCommand\mathstylel{m}{#1}
\NewDocumentCommand\mathstylelth{m}{#1}
\NewDocumentCommand\mathstylen{m}{#1}
\NewDocumentCommand\mathstyleod{m}{#1}
\NewDocumentCommand\mathstylereal{m}{#1}
\NewDocumentCommand\mathstylerel{m}{#1}
\NewDocumentCommand\mathstyles{m}{#1}
\NewDocumentCommand\mathstyletop{m}{\symfrak{#1}}

\NewObject\mClassMath\algA{\mathstylealg{A}}
\NewObject\mClassMath\algB{\mathstylealg{B}}
\NewObject\mClassMath\algC{\mathstylealg{C}}
\NewObject\mClassMath\algFree{\mathstylealg{Fr}}

\DeclareObject\mClassMath\aa{\mathstylea{a}} \NewObject\mClassMath\ab{\mathstylea{b}}
\NewObject\mClassMath\ac{\mathstylea{c}}

\DeclareObject\mClassMath\cdkappa{\mathstylecd{\kappa}}
\DeclareObject\mClassMath\cdlambda{\mathstylecd{\lambda}}

\NewObject\mClassMath\fe{\mathstylef{e}}
\NewObject\mClassMath\ff{\mathstylef{f}}
\NewObject\mClassMath\fg{\mathstylef{g}}
\NewObject\mClassMath\fh{\mathstylef{h}}
\NewObject\mClassMath\fk{\mathstylef{k}}
\NewObject\mClassMath\fl{\mathstylef{l}}
\NewObject\mClassMath\fp{\mathstylef{p}}
\NewObject\mClassMath\fq{\mathstylef{q}}
\NewObject\mClassMath\fr{\mathstylef{r}}
\NewObject\mClassMath\fs{\mathstylef{s}}
\NewObject\mClassMath\ft{\mathstylef{t}}
\NewObject\mClassMath\fu{\mathstylef{u}}
\NewObject\mClassMath\fiota{\mathstylef{\iota}}
\NewObject\mClassMath\fsigma{\mathstylef{\sigma}}
\NewObject\mClassMath\fSeqShift{\mathstylef{\sigma}}
\SemantexRecordObject{\fComp}
\NewDocumentCommand{\fComp}{m}{
    \SemantexRecordSource { \fComp { #1 } }
    \paren { #1 }
}
\SemantexRecordObject{\fDom}
\NewDocumentCommand{\fDom}{m}{
    \SemantexRecordSource { \fDom { #1 } }
    \operatorname{dom}\paren{ #1 }
}
\SemantexRecordObject{\fIm}
\NewDocumentCommand{\fIm}{m}{
    \SemantexRecordSource { \fIm { #1 } }
    \operatorname{Im}\paren{ #1 }
}
\SemantexRecordObject{\fSupp}
\NewDocumentCommand{\fSupp}{m}{
    \SemantexRecordSource { \fSupp { #1 } }
    \operatorname{supp}\paren{ #1 }
}

\NewObject\mClassMath\fVal{\symfrak{V}}
\NewObject\mClassMath\fSpan{\mathrm{sp}}
\NewObject\mClassMath\fTy{\mathrm{ty}}
\NewObject\mClassMath\fClust{\mathrm{clu}}

\NewObject\mClassMath\frF{\mathstylefr{F}}
\NewObject\mClassMath\frG{\mathstylefr{G}}
\NewObject\mClassMath\frH{\mathstylefr{H}}
\NewObject\mClassMath\frT{\mathstylefr{T}}

\NewObject\mClassMath\frlQ{\mathstylefrl{Q}}
\NewObject\mClassMath\frlS{\mathstylefrl{S}}
\NewObject\mClassMath\frlT{\mathstylefrl{T}}

\NewObject\mClassMath\frmM{\mathstylefrm{M}}
\NewObject\mClassMath\frmN{\mathstylefrm{N}}

\NewObject\mClassMath\lBox{\lbox}
\NewObject\mClassMath\lDiamond{\ldiamond}
\NewObject\mClassMath\lNext{\lcircle}

\NewObject\mClassMath\logLambda{\mathstylelog{\Lambda}}
\NewObject\mClassMath\logForm{\mathrm{Fm}}
\NewObject\mClassMath\logK{\mathbf{K}}
\NewObject\mClassMath\logKWeakTrans{\mathbf{wK4}}
\NewObject\mClassMath\logKWeakTransLin{\mathbf{wK4.3}}
\NewObject\mClassMath\logKTrans{\mathbf{K4}}
\NewObject\mClassMath\logKTransLin{\mathbf{K4.3}}
\SemantexRecordObject{\logKTransBoundedWidth}
\DeclareDocumentCommand{\logKTransBoundedWidth}{m}{\SemantexRecordSource {\logKTransBoundedWidth {#1}}\UseClassInCommand\mClassMath{\mathbf{K4BW}_{#1}}}
\NewObject\mClassMath\logKRefl{\mathbf{KT}}
\NewObject\mClassMath\logKReflTrans{\mathbf{S4}}
\NewObject\mClassMath\logKReflTransLin{\mathbf{S4.3}}
\NewObject\mClassMath\logKReflTransLinFinal{\mathbf{S4.3.1}}
\NewObject\mClassMath\logKReflTransConfl{\mathbf{S4.2}}
\NewObject\mClassMath\logKReflTransConflFinal{\mathbf{S4.2.1}}
\NewObject\mClassMath\logKReflTransFinal{\mathbf{S4.1}}
\SemantexRecordObject{\logKReflTransBoundedWidth}
\DeclareDocumentCommand{\logKReflTransBoundedWidth}{m}{\SemantexRecordSource {\logKReflTransBoundedWidth {#1}}\UseClassInCommand\mClassMath{\mathbf{S4BW}_{#1}}}
\NewObject\mClassMath\logGrz{\mathbf{Grz}}
\NewObject\mClassMath\logGrzConfl{\mathbf{Grz.2}}
\NewObject\mClassMath\logGrzLin{\mathbf{Grz.3}}
\SemantexRecordObject{\logGrzBoundedWidth}
\DeclareDocumentCommand{\logGrzBoundedWidth}{m}{\SemantexRecordSource {\logGrzBoundedWidth {#1}}\UseClassInCommand\mClassMath{\mathbf{GrzBW}_{#1}}}
\NewObject\mClassMath\logGl{\mathbf{GL}}
\NewObject\mClassMath\logGlLin{\mathbf{GL.3}}
\SemantexRecordObject{\logGlBoundedWidth}
\DeclareDocumentCommand{\logGlBoundedWidth}{m}{\SemantexRecordSource {\logGlBoundedWidth {#1}}\UseClassInCommand\mClassMath{\mathbf{GLBW}_{#1}}}

\NewObject\mClassMath\lphi{\mathstylel{\varphi}}
\NewObject\mClassMath\lpsi{\mathstylel{\psi}}
\NewObject\mClassMath\lchi{\mathstylel{\chi}}

\NewObject\mClassMath\lap{\mathstylel{p}}
\NewObject\mClassMath\laq{\mathstylel{q}}

\NewObject\mClassMath\lthS{\mathstylelth{S}}
\NewObject\mClassMath\lthT{\mathstylelth{T}}
\NewObject\mClassMath\lthGamma{\mathstylelth{\Gamma}}
\NewObject\mClassMath\lthDelta{\mathstylelth{\Delta}}
\NewObject\mClassMath\lthSigma{\mathstylelth{\Sigma}}
\NewObject\mClassMath\lthSubform{\mathrm{SubF}}\NewObject\mClassMath\lthProp{\symbb{P}}

\NewObject\mClassMath\nii{\mathstylen{i}}
\NewObject\mClassMath\nj{\mathstylen{j}}
\NewObject\mClassMath\nk{\mathstylen{k}}
\NewObject\mClassMath\nl{\mathstylen{l}}
\NewObject\mClassMath\nm{\mathstylen{m}}
\NewObject\mClassMath\nn{\mathstylen{n}}
\NewObject\mClassMath\nN{\mathstylen{N}}
\NewObject\mClassMath\np{\mathstylen{p}}

\NewObject\mClassMath\odalpha{\mathstyleod{\alpha}}
\NewObject\mClassMath\odbeta{\mathstyleod{\beta}}
\NewObject\mClassMath\odgamma{\mathstyleod{\gamma}}
\NewObject\mClassMath\oddelta{\mathstyleod{\delta}}
\NewObject\mClassMath\odlambda{\mathstyleod{\lambda}}

\NewObject\mClassMath\realr{\mathstylereal{r}}

\NewObject\mClassMath\relD{\mathstylerel{D}}
\NewObject\mClassMath\relE{\mathstylerel{E}}
\NewObject\mClassMath\relF{\mathstylerel{F}}
\NewObject\mClassMath\relG{\mathstylerel{G}}
\NewObject\mClassMath\relR{\mathstylerel{R}}
\NewObject\mClassMath\relS{\mathstylerel{S}}
\NewObject\mClassMath\relZ{\mathstylerel{Z}}
\NewObject\mClassMath\relEquiv{\equiv}
\NewObject\mClassMath\relPreEquiv{\sqsubseteq}
\SemantexRecordObject{\relComp}
\NewDocumentCommand{\relComp}{m}{
    \SemantexRecordSource { \relComp { #1 } }
    \paren { #1 }
}
\NewObject\mClassMath\relBisim{\rightleftharpoons}
\NewObject\mClassMath\relDiag{\Delta}

\NewObject\mClassMath\sEmpty{\emptyset}
\NewObject\mClassMath\sA{\mathstyles{A}}
\NewObject\mClassMath\sB{\mathstyles{B}}
\NewObject\mClassMath\sC{\mathstyles{C}}
\NewObject\mClassMath\sD{\mathstyles{D}}
\NewObject\mClassMath\sE{\mathstyles{E}}
\NewObject\mClassMath\sF{\mathstyles{F}}
\NewObject\mClassMath\sG{\mathstyles{G}}
\NewObject\mClassMath\sH{\mathstyles{H}}
\NewObject\mClassMath\sI{\mathstyles{I}}
\NewObject\mClassMath\sJ{\mathstyles{J}}
\NewObject\mClassMath\sP{\mathstyles{P}}
\NewObject\mClassMath\sS{\mathstyles{S}}
\NewObject\mClassMath\sT{\mathstyles{T}}
\NewObject\mClassMath\sU{\mathstyles{U}}
\NewObject\mClassMath\sV{\mathstyles{V}}
\NewObject\mClassMath\sW{\mathstyles{W}}
\NewObject\mClassMath\sX{\mathstyles{X}}
\NewObject\mClassMath\sY{\mathstyles{Y}}
\NewObject\mClassMath\sZ{\mathstyles{Z}}
\NewObject\mClassMath\sTwo{\symbb{2}}
\NewObject\mClassMath\sTypes{\mathrm{Ty}}
\NewObject\mClassMath\sTopClust{\mathrm{TopClust}}
\NewObject\mClassMath\sStab{\mathrm{Stab}}

\NewObject\mClassMath\sPowerSet{\mathcal{P}}

\NewObject\mClassMath\topF{\mathstyletop{F}}
\NewObject\mClassMath\topG{\mathstyletop{G}}
\NewObject\mClassMath\topX{\mathstyletop{X}}
\NewObject\mClassMath\topY{\mathstyletop{Y}}
\NewObject\mClassMath\topZ{\mathstyletop{Z}}

\NewObject\mClassMath\vi{i}
\NewObject\mClassMath\vr{r}
\NewObject\mClassMath\vs{s}
\NewObject\mClassMath\vt{t}
\NewObject\mClassMath\vu{u}
\NewObject\mClassMath\vv{v}
\NewObject\mClassMath\vw{w}
\NewObject\mClassMath\vx{x}
\NewObject\mClassMath\vy{y}
\NewObject\mClassMath\vz{z}

\NewObject\mClassMath\Naturals{\symbb{N}}
\NewObject\mClassMath\Integers{\symbb{Z}}
\NewObject\mClassMath\Rationals{\symbb{Q}}
\NewObject\mClassMath\Reals{\symbb{R}}
\NewObject\mClassMath\ExtReals{\overline{\symbb{R}}}

\NewSymbolClass\mClassOperators

\NewObject\mClassOperators\fBinProd{\times}

\NewObject\mClassOperators\sMinus{\setminus}
\NewObject\mClassOperators\sCartProd{\times}
\NewObject\mClassOperators\sUnion{\cup}
\NewObject\mClassOperators\sDisjointUnion{\amalg}
\NewObject\mClassOperators\sIntersection{\cap}
\NewObject\mClassOperators\sSetUnion{\bigcup\:}
\NewObject\mClassOperators\sSetDisjointUnion{\coprod\:}
\NewObject\mClassOperators\sSetIntersection{\bigcap\:}
\NewObject\mClassOperators\sLimitUnion{\bigcup}
\NewObject\mClassOperators\sLimitDisjointUnion{\coprod}
\NewObject\mClassOperators\sLimitIntersection{\bigcap}
\NewObject\mClassOperators\sSubsetEq{\subseteq}
\NewObject\mClassOperators\sSubsetStrict{\subsetneq}
\NewObject\mClassOperators\sSupersetEq{\supseteq}
\NewObject\mClassOperators\sSupersetStrict{\supsetneq}
\NewObject\mClassOperators\sEquinumerous{\sim}

\NewObject\mClassOperators\Equal{=}
\NewObject\mClassOperators\LessStrict{<}
\NewObject\mClassOperators\LessEq{\leq}
\NewObject\mClassOperators\GreaterStrict{>}
\NewObject\mClassOperators\GreaterEq{\geq}

\NewObject\mClassOperators\join{\vee}
\NewObject\mClassOperators\meet{\wedge}
\NewObject\mClassOperators\JOIN{\bigvee}
\NewObject\mClassOperators\MEET{\bigwedge}

\NewObject\mClassOperators\lProves{\vdash}
\NewObject\mClassOperators\lNotProves{\nvdash}
\NewObject\mClassOperators\lValidates{\vDash}
\NewObject\mClassOperators\lNotValidates{\nvDash}
\NewObject\mClassOperators\Isomorphic{\cong}
\NewObject\mClassOperators\NotIsomorfic{\ncong}
\NewObject\mClassOperators\ElementaryEquiv{\equiv}

\NewObject\mClassOperators\Sum{\sum}

\NewObject\mClassOperators\vequiv{\sim}

\SetupClass\mClassDelim{
    output=\mClassMath,
    define keys={
        {prime}{output options={prime}},
        {'}{output options={'}},
        {''}{output options={''}},
        {'''}{output options={'''}},
        {fn}{set arg single keys={\cdot}},
    },
    define keys[1]={
        {default}{output options={default={#1}}},
        {lower}{output options={lower={#1}}},
    },
}

\SetupClass\mClassMath{
    output=\mClassMath,
    define keys={
        {vec}{command=\vec},
        {tilde}{command=\tilde},
        {widetilde}{command=\widetilde},
        {hat}{command=\hat},
        {star}{upper=\ast},
        {overline}{command=\overline},
        {box}{lower={\lBox}},
        {converse}{upper=\mathrm{op}},
        {inv}{upper={-1}},
        {preimage}{upper={-1}},
        {gen}{overline},
        {complex alg}{upper=\ast},
        {negative}{upper={\leq 0}},
        {Delta}{upper={\Delta}},
        {refl closure}{overline},
        {trans closure}{upper=+},
        {refl trans closure}{star},
        {irrefl closure}{upper=\circ},
{domain}{return, command=\fDom},
        {dom}{domain},
{image}{return, command=\fIm},
{card}{return, command=\card},
{F}{lower={\mathrm{f}}},
{P}{lower={\mathrm{p}}},
{M}{lower={\mathrm{m}}},
    },
    define keys[1]={
        {pow}{upper=#1},
        {restrict}{right return,symbol put right={\mathopen{}\upharpoonright\mathclose{} {#1}}},
        {quotient}{right return,symbol put right={\mathopen{}\mathslash\mathclose{} {#1}}},
    },
}

\SetupClass\mClassOperators{
    define keys={
        {overline}{command=\overline},
    },
    define keys[1]={
        {default}{output options={default={#1}}},
        {lower}{output options={lower={#1}}},
        {upper}{output options={upper={#1}}},
    },
}

\ExplSyntaxOn

\NewDocumentCommand{\ElimCrit}{o} {
    \IfValueTF { #1 }
    {
        \group_begin:
\keys_set:nn { mymath/elimcrit } { #1 }
\symcal{E}
\c_math_subscript_token { \l__mymath_elimcrit_predefined_cs:n { \tl_use:N \l__mymath_elimcrit_theory_tl } }
\tl_if_empty:NTF \l__mymath_elimcrit_upset_tl {} {
            \paren{
\tl_if_empty:NTF \l__mymath_elimcrit_model_tl {} {
                \tl_use:N \l__mymath_elimcrit_model_tl ,
            }
\tl_use:N \l__mymath_elimcrit_upset_tl
            }
        }
        \group_end:
    }
    { \symcal{E} }
}

\keys_define:nn { mymath/elimcrit }
    {
        model .tl_set:N = \l__mymath_elimcrit_model_tl,
        model .value_required:n = true,
        upset .tl_set:N = \l__mymath_elimcrit_upset_tl,
        upset .value_required:n = true,
        theory .tl_set:N = \l__mymath_elimcrit_theory_tl,
        theory .value_required:n = true,

        predefined .choice:,
        predefined / none .code:n = \cs_set_protected_nopar:Nn \l__mymath_elimcrit_predefined_cs:n {##1},
        predefined / weak .code:n = \cs_set_protected_nopar:Nn \l__mymath_elimcrit_predefined_cs:n {##1\textrm{-w}},
        predefined / strong .code:n = \cs_set_protected_nopar:Nn \l__mymath_elimcrit_predefined_cs:n {##1\textrm{-s}},
        predefined .initial:n = none,

        weak .meta:n = {predefined=weak},
        weak .value_forbidden:n = true,
        strong .meta:n = {predefined=strong},
        strong .value_forbidden:n = true,

        m .meta:n = {model=#1},
        u .meta:n = {upset=#1},
        t .meta:n = {theory=#1},
    }

\ExplSyntaxOff

\RequirePackage{tikz-cd}
\usetikzlibrary{calc}
\usetikzlibrary{decorations.pathmorphing}
\usetikzlibrary{spath3}

\tikzset{curve/.style={settings={#1},to path={(\tikztostart)
    .. controls ($(\tikztostart)!\pv{pos}!(\tikztotarget)!\pv{height}!270:(\tikztotarget)$)
    and ($(\tikztostart)!1-\pv{pos}!(\tikztotarget)!\pv{height}!270:(\tikztotarget)$)
    .. (\tikztotarget)\tikztonodes}},
    settings/.code={\tikzset{quiver/.cd,#1}
        \def\pv##1{\pgfkeysvalueof{/tikz/quiver/##1}}},
    quiver/.cd,pos/.initial=0.35,height/.initial=0}

\tikzset{between/.style n args={2}{/tikz/spath/at end path construction={
    \tikzset{spath/split at keep middle={current}{#1}{#2}}
}}}

\tikzset{tail reversed/.code={\pgfsetarrowsstart{tikzcd to}}}
\tikzset{2tail/.code={\pgfsetarrowsstart{Implies[reversed]}}}
\tikzset{2tail reversed/.code={\pgfsetarrowsstart{Implies}}}
\tikzset{no body/.style={/tikz/dash pattern=on 0 off 1mm}}

 \usetikzlibrary{arrows.meta}
\tikzset{ast/.tip={Glyph[glyph math command=astsmall, glyph length=0.9ex,glyph axis=2.05pt]}}

\SetupObject\relEquiv{
    define keys={
        {s}{lower=\sigma},
        {d}{lower=\delta},
        {heavy lower slot}{lower=\smblkcircle},
    },
    define keys[1]={
        {arg}{return, command={\mathord}, set arg single keys={#1}},
        {ty}{lower={#1\textrm{-ty}}},
        {span}{lower={#1\textrm{-sp}}},
    },
}

\NewObject\mClassDelim\sEquivClass[left par=\lbrack, right par=\rbrack]
\SetupObject\sEquivClass{
    define keys={
        {s}{lower=\sigma},
        {d}{lower=\delta},
        {heavy lower slot}{lower=\smblkcircle},
    },
    define keys[1]={
        {ty}{lower={#1\textrm{-ty}}},
        {span}{lower={#1\textrm{-sp}}},
    },
}

\SetupObject\fVal{
    symbol={V},
}

\SetupObject\sTopClust{
    symbol={\mathrm{Top}},
    define keys[1]={
        {default}{arg={#1}},
    },
}

\SetupObject\lBox{
    symbol={\mdlgwhtsquare},
    define keys={
        {refl closure}{symbol=\boxdot},
    },
}

\newlist{proofcaseitem}{itemize}{1}
\setlist[proofcaseitem]{label=\textendash}
 
\newcommand*{\theauthors}{Simon Santschi and Niels C.~Vooijs}
\newcommand*{\thetitle}{Logics Containing wK4: Selection à la Fine}
\newcommand*{\thedate}{\today}

\author{
Simon Santschi\thanks{Supported by the Swiss National Science Foundation (SNSF), grant no. 200021\textunderscore215157.}
\institute{Mathematical Institute\\
University of Bern\\
Bern, Switzerland}
\email{simon.santschi@unibe.ch}
\and
Niels C.~Vooijs\footnotemark[1]
\institute{Mathematical Institute\\
University of Bern\\
Bern, Switzerland}
\email{ncvooijs@gmail.com}
}
\date{\thedate}
\title{\thetitle}

\newcommand{\titlerunning}{\thetitle}
\newcommand{\authorrunning}{\theauthors}

\begin{document}
    \maketitle

    \begin{abstract}
We generalize Fine's Iterative Selection Method to the weakly transitive setting.
        In particular, this provides a transparent frame-theoretic proof of the finite model property for (strongly) cofinal subframe logics extending \(\logKWeakTrans\), which was previously established using algebraic methods.
Using the same construction, we generalize Fine's Finite Width Theorem to the weakly transitive setting,
connecting these two celebrated theorems of Fine.
\end{abstract}

    \section{Introduction}
    A well-known theorem of Fine~\cite{Fine1985-logics-containing-K4-part-2} states that every subframe logic extending \(\logKTrans\) has the finite model property (or \fmp{} for short).
Here, a normal modal logic \(\logLambda\) is called \emph{subframe} \tdefiff{} every definable submodel of a model of \(\logLambda\) is itself a model of \(\logLambda\).
Fine~\cite{Fine1985-logics-containing-K4-part-2} gives two proofs for this fact.
The first proof \cite[Theorem~4.5]{Fine1985-logics-containing-K4-part-2} relies on subframe formulas and a relatively straightforward selection on the canonical model.
The second proof \cite[Theorem~6.6]{Fine1985-logics-containing-K4-part-2} is purely model-theoretic but relies on a more involved iterated selection-quotient procedure\footnote{Fine~\cite{Fine1985-logics-containing-K4-part-2} does not mention the definability of the selected subset, and therefore states the result for Kripke complete logics \cite[Corollary~6.3]{Fine1985-logics-containing-K4-part-2}. The fact that the selected subset is definable was noted by Kracht~\cite{Kracht1990-internal-definability-and-completeness-in-modal-logic}.}.

This second approach has several advantages.
As Fine puts it: \blockcquote[Section~6, \citepage{638}]{Fine1985-logics-containing-K4-part-2}{the extra work involved in the direct proof is not wasted; for it enables us to extract further information about subreduct models}.
Moreover, this approach generalizes the result from subframe logics to cofinal subframe logics, which would only later be introduced by Zakharyaschev~\cite{Zakharyaschev1992-canonical-formulas-for-K4-part-1}.
Finally, Zakharyaschev's theory of canonical formulas builds upon this construction \cite{Zakharyaschev1992-canonical-formulas-for-K4-part-1,Zakharyaschev1996-canonical-formulas-for-K4-part-2,Zakharyaschev1997-canonical-formulas-for-K4-part-3}.

We generalize Fine's Iterative Selection Method to the weakly transitive setting, obtaining the following model-theoretic analogue to \cite[Lemma~5.1]{BezhanishviliGhilardiJibladze2011-an-algebraic-approach-to-subframe-logics-modal-case}.

\begin{maintheorem}\zlabel{mthm:selection}
    Let \(\frmM\) be a weakly transitive Kripke model and \(\lthSigma\) a finite subformula-closed set of formulas.
	Then there exists a definable \(\lthSigma\)-sub-p-morphism from \(\frmM\) onto a finite model.
\end{maintheorem}

Here a \(\lthSigma\)-sub-p-morphism is a generalization of Fine's \(\nn\)-subreduction, see \zcref{sec:setting-up}.
As an immediate corollary, we obtain a model-theoretic proof of the following theorem, which was already proved algebraically, see \cite[Theorem~5.9]{BezhanishviliGhilardiJibladze2011-an-algebraic-approach-to-subframe-logics-modal-case} and \cite[Theorem~7.14]{BezhanishviliBezhanishvili2012-canonical-formulas-wK4}.

\begin{maintheorem}\zlabel{mthm:fmp}
    Every strongly cofinal subframe logic extending \(\logKWeakTrans\) has the \fmp{}.
\end{maintheorem}

Moreover, we generalize the construction beyond finite filtration sets, by extending the iteration to the transfinite.
This allows us to generalize, using the selection construction, Fine's Finite Width Theorem \cite{Fine1974-logics-containing-K4-part-1} to the weakly transitive setting.

\begin{maintheorem}\zlabel{mthm:finite-width}
    Every logic of finite width extending \(\logKWeakTrans\) is complete w.r.t.\@ its conversely pre-well-founded frames, and is therefore, in particular, Kripke complete.
\end{maintheorem}

That is, we obtain two of Fine's celebrated theorems, generalized to the weakly transitive setting, from a single generalized construction, connecting the two results.

Compared to \cite[Lemma~5.1]{BezhanishviliGhilardiJibladze2011-an-algebraic-approach-to-subframe-logics-modal-case}, \zcref{mthm:selection} is a very slight generalization in that it gives some more guarantees about the selected subset.
The important corollary \zcref{mthm:fmp} follows equally from both constructions, where it should be noted that \emph{strong} cofinality was introduced only in \cite{BezhanishviliBezhanishvili2012-canonical-formulas-wK4}, and is thus not stated in \cite[Theorem~5.9]{BezhanishviliGhilardiJibladze2011-an-algebraic-approach-to-subframe-logics-modal-case} but does easily follow from their construction, as noted in \cite{BezhanishviliBezhanishvili2012-canonical-formulas-wK4}.
Thus, our construction can be understood as a purely model-theoretic alternative to \cite{BezhanishviliGhilardiJibladze2011-an-algebraic-approach-to-subframe-logics-modal-case}, with the added benefit of uniting it with the Finite Width Theorem in a single construction.

Another strand in the literature, starting with Fine's \cite{Fine1985-logics-containing-K4-part-2} first \fmp{} proof, is to perform selection on canonical models.
These, and more generally descriptive models, are easier to deal with for selection methods, since for every point \(\vx\) that satisfies a formula \(\lphi\), there exists a \emph{maximal} successor of \(\vx\) that satisfies \(\lphi\).
The selected subset can therefore be defined directly without the need for an iterated construction. A canonicity assumption on the logic then also obliviates the need for definability of the selection.
Examples of this approach in the weakly transitive setting are \cite{BaltagBezhanishviliFernandez2023-topological-mu-calclulus-completeness-and-decidability,KudinovShapirovsky2025-two-types-of-filtrations-for-wK4-and-relatives}.

The logic \(\logKWeakTrans = \logK \oplus \paren{\lap \land \lBox\lap} \limplies \lBox\lBox\lap\) is the logic of weakly transitive frames \cite{Esakia2001-weak-transitivity}.
This logic and its extensions have recently attracted attention \cite{BezhanishviliGhilardiJibladze2011-an-algebraic-approach-to-subframe-logics-modal-case,BezhanishviliBezhanishvili2012-canonical-formulas-wK4,BezhanishviliEsakiaGabelaia2011-spectral-and-T0-spaces-in-d-semantics,ChenMa2025-the-McKinsey-axiom-on-weakly-transitive-frames} for two reasons.
First, \(\logKWeakTrans\) is the logic of topological spaces when \(\lDiamond\) is interpreted as the Cantor derivative.
This so-called \emph{topological d-semantics} was introduced by McKinsey and Tarski~\cite{McKinseyTarski1944-the-algebra-of-topology} along with the classical \emph{closure semantics}, but completeness of \(\logKWeakTrans\)  with respect to this semantics was only proved in \cite{Esakia2001-weak-transitivity}.
Second, many results for extensions of \(\logKTrans\) generalize to extensions of \(\logKWeakTrans\), but the proofs often require non-trivial modification.

\paragraph{Outline.}
The paper is organized as follows.
In \zcref{sec:prelim}, we recall the basic notions and notations needed in the paper.
In \zcref{sec:setting-up}, we generalize Fine's \(\nn\)-subreductions to the weakly transitive setting and introduce auxiliary notions used in our selection construction, which is described in \zcref{sec:selection-construction}.
\zcref[S]{sec:covering,sec:diagram-constr} are concerned with the termination of the construction and the definability of the selected subset, respectively. Both results rely on suitably defined bisimulations.
In \zcref{sec:Fines-thms}, we combine the results from the earlier sections to derive \zcref{mthm:selection-again,mthm:fmp-again,mthm:finite-width-again}.
Finally, in \zcref{sec:conclusion}, we
discuss possible directions for further research.
 
    \section{Preliminaries}\zlabel{sec:prelim}
    We assume familiarity with basic Kripke semantics for normal unimodal logic, see, e.g.,\@ \cite{CZ1997,BdRV2001}.
By a \emph{logic} we always mean a normal unimodal logic.
We write \(\lBox[refl closure]\lphi \coloneqq \lphi \land \lBox\lphi\).

Let  \(\frF = \structuple{\sX, \relR}\) be a (Kripke) frame. For \(\vx \in \sX\), we write \(\relR{\vx} \coloneqq \Set{\vy \in \sX}[\homotuple{\vx, \vy} \in \relR]\).
For \(\vy \in \relR{\vx}\) we say \(\vx\) \emph{sees} \(\vy\) and \(\vy\) is a \emph{successor} of \(\vx\).
If, moreover, \(\vx \notin \relR{\vy}\), then we call \(\vy\) a \emph{strict successor} of \(\vx\).
We write \(\relR[converse] \coloneqq \Set{\homotuple{\vy, \vx}}[\homotuple{\vx, \vy} \in \relR]\) for the converse of \(\relR\), and \(\relR[refl trans closure]\) for the reflexive transitive closure of \(\relR\).
Given a set \(\sY \sSubsetEq \sX\) and \(\vx \in \sX\), we write \(\fSpan[\sY]{\vx} \coloneqq \relR{\vx} \sIntersection \sY\) for the \(\sY\)-\emph{span} of \(\vx\), \(\fSpan[refl trans closure, \sY]{\vx} \coloneqq \relR[refl trans closure]{\vx} \sIntersection \sY\), and \(\relEquiv[span=\sY]\) for the equivalence relation on \(\sX\) induced by \(\fSpan[\sY]\).
For any equivalence relation \(\relEquiv[heavy lower slot]\), we write \(\sEquivClass[heavy lower slot]{\vx}\) for the equivalence class of \(\vx\).
For a binary relation \(\relS\), we denote its \emph{domain} by \(\relS[dom] \coloneqq \Set{\vx}[\homotuple{\vx,\vy}\in \relS]\).
We apply logical connectives to subsets of a frame, e.g.,\@ \(\lDiamond\sY = \relR[converse]{\sY}\) and \(\sY \limplies \sZ = \lnot\sY \sUnion \sZ\).

Recall that \(\frF\) is \emph{weakly transitive} \tdefiff{} whenever \(\vy \in \relComp{\relR \circ \relR}{\vx}\), then either \(\vy \in \relR{\vx}\) or \(\vy = \vx\).
In this case, \(\relR[refl trans closure]\) coincides with the reflexive closure of \(\relR\).
For the rest of the paper, we will only be concerned with weakly transitive frames.

A \emph{cluster} is a maximal set \(\sC\) such that \(\sC \sCartProd \sC \sSubsetEq \relR[refl trans closure]\).
Every point is a member of exactly one cluster.
A set \(\sU \sSubsetEq \sX\) is called an \emph{upset} or \emph{upwards closed} \tdefiff{} \(\relR{\sU} \sSubsetEq \sU\).
We say that \(\sY \sSubsetEq \sX\) \emph{covers} \(\sY['] \sSubsetEq \sX\) \tdefiff{} \(\sY['] \sSubsetEq \relR[refl trans closure, spar, converse]{\sY}\).
It is \emph{cofinal} \tdefiff{} it covers \(\sX\) and \emph{strongly cofinal} \tdefiff{}, moreover, \(\relComp{\relR \circ \relR}[converse]{\sY} \sSubsetEq \relR[converse]{\sY}\).

For a model \(\frmM = \structuple{\sX,\relR,\fVal}\) on the frame \(\frF\) and a formula \(\lphi\), we write \(\sem{\lphi}[\frmM]\) for the set of points \(\vx\) satisfying \(\lphi\), i.e.\@ \(\structuple{\frmM, \vx} \lValidates \lphi\).
Given a set of formulas \(\lthSigma\), we write \(\fTy[\lthSigma]{\frmM, \vx} \coloneqq \Set{\lphi \in \lthSigma}[\vx \in \sem{\lphi}[\frmM]]\) for the \(\lthSigma\)-\emph{type} of \(\vx\).
In both cases we omit \(\frmM\) from the notation when it is clear from the context.
We write \(\relEquiv[ty=\lthSigma]\) for the equivalence relation on \(\sX\) induced by \(\fTy[\lthSigma]\).
We say \(\frmM\) \emph{validates} \(\lphi\) \tdefiff{} \(\sem{\lphi}[\frmM] = \sX\), and \(\frmM\) \emph{is a model of} a logic \(\logLambda\) or is a \emph{\(\logLambda\)-model} \tdefiff{} it validates every \(\lphi \in \logLambda\).
A logic \(\logLambda\) is said to have the \emph{finite model property} (or \emph{\fmp{}} for short) \tdefiff{} it is complete with respect to its finite models.

A set \(\sY \sSubsetEq \sX\) is called \emph{definable} in \(\frmM\) \tdefiff{} there exists a formula \(\lphi\) such that \(\sY = \sem{\lphi}[\frmM]\), and \(\vx \in \sX\) is \emph{definable} \tdefiff{} \(\Set{\vx}\) is definable.
A \emph{definable variant} of \(\frmM\) is a model \(\frmM[']\) on \(\frF\) such that \(\sem{\lap}[\frmM[']]\) is definable in \(\frmM\) for every atomic proposition \(\lap\).
In this case, if \(\frmM\) is a \(\logLambda\)-model for some logic \(\logLambda\), then \(\frmM[']\) is a \(\logLambda\)-model as well.
\(\frmM\) is called \emph{differentiated} \tdefiff{} for all \(\vx, \vy \in \sX\), if \(\vx \neq \vy\), then there exists a definable set \(\sY\) with \(\vx \in \sY\) and \(\vy \notin \sY\), and \emph{atomic} \tdefiff{} every point is definable.

Recall that a \emph{p-morphism} between models is a surjection which is also a bisimulation.
Any bisimulation equivalence \(\relEquiv\) on \(\frmM\) induces a p-morphism from \(\frmM\) to the quotient \(\frmM[quotient=\relEquiv]\).
A submodel \(\frmM[']\) of \(\frmM\) is the restriction of \(\frmM\) to some set \(\sY \sSubsetEq \sX\), and a p-morphism from \(\frmM[']\) to \(\frmN\) is called a \emph{sub-p-morphism} from \(\frmM\) to \(\frmN\).
Both the submodel and the sub-p-morphism are called \emph{cofinal}, \emph{definable}, etc.\@ \tdefiff{} \(\sY\) has the property in question in \(\frmM\).
If \(\frmM[']\) is a definable submodel of \(\frmM\), then every set definable in \(\frmM[']\) is definable in \(\frmM\).

A logic \(\logLambda\) is called \emph{subframe}, \emph{cofinal subframe}, or \emph{strongly cofinal subframe}, respectively, \tdefiff{} for every \(\logLambda\)-model \(\frmM\), every definable, definable cofinal, or definable strongly cofinal submodel of \(\frmM\), respectively, is itself a \(\logLambda\)-model.

We say that a frame \(\frF\) or model \(\frmM\) has \emph{finite width} \tdefiff{} every anti-chain in it is finite, and a logic (extending \(\logKWeakTrans\)) is of \emph{finite width} \tdefiff{} it is complete w.r.t.\@ its models of finite width.
Note that we do not require a fixed finite bound on the size of anti-chains.
We call \(\frF\) \emph{conversely pre-well-founded} \tdefiff{} there is no infinite sequence \(\Function{\ff}{\Naturals}{\sX}\) in it such that \(\ff{\nn + 1}\) is a strict successor of \(\ff{\nn}\) for every \(\nn \in \Naturals\).
The frame \(\frF\) has the \emph{finite cover property} \tdefiff{} every set \(\sY \sSubsetEq \sX\) is covered by a finite subset \(\sY['] \sSubsetEq \sY\).
The finite cover property is equivalent to
the combination of the finite width and converse pre-well-foundedness properties,
and is also equivalent to the statement that for every infinite decreasing sequence of upsets \(\paren{\sU[\nn]}[\nn \in \Naturals]\), there exists \(\nii \in \Naturals\) such that \(\sU[\nii] = \sU[\nii + 1]\).
 
    \section{Towards the Selection Construction}\zlabel{sec:setting-up}
    In this section, we generalize Fine's \(\nn\)-subreductions and Zakhariaschev's \(\lthSigma\)-subreductions to the weakly transitive setting, and set up some further definitions we will need for the selection construction in the next section.
Let us fix a subformula-closed set of formulas \(\lthSigma\), and denote by \(\lthProp\) the set of atomic propositions in \(\lthSigma\).
Moreover, fix, for the rest of this section, a model \(\frmM\) on a weakly transitive frame \(\frF = \structuple{\sX, \relR}\).

\begin{definition}[Eliminability]
	Let \(\sY \sSubsetEq \sX\).
	Then \(\vx \in \sX \sMinus \sY\) is called
	\emph{\(\lthSigma\)-eliminable w.r.t.\@ \(\sY\)}
	iff for every \(\lphi \in \lthSigma\), there exists \(\vy \in \fSpan[\sY]{\vx}\) with \(\vx \in \sem{\lphi} \iff \vy \in \sem{\lphi}\).
	We write \(\ElimCrit[t=\lthSigma, m=\frmM, u=\sY]\) for the set of \(\lthSigma\)-eliminable points w.r.t.\@ \(\sY\), and generally omit \(\lthSigma\) from the notation.
	A point \(\vx \in \sX\) is called \emph{\(\lthSigma\)-eliminable} \tdefiff{} it is \(\lthSigma\)-eliminable w.r.t.\@ \(\Set{\vy \in \relR{\vx}}[\vx \notin \relR{\vy}]\).
\end{definition}

\begin{lemma}\zlabel{lemma:eliminability:definable}
    If \(\lthSigma\) is finite and \(\sY\) is definable in \(\frmM\), then so is \(\ElimCrit[t=\lthSigma, m=\frmM, u=\sY]\).
\end{lemma}
\begin{proof}
    Let \(\symcal{A}\) be the set containing \(\sem{\lphi}\) and \(\sem{\lnot\lphi}\) for each \(\lphi \in \lthSigma\).
    Then
	\begin{equation*}
		\ElimCrit[t=\lthSigma, m=\frmM, u=\sY]
		=
		\lnot\sY \sIntersection \sSetIntersection\Set{ \sA \limplies \lDiamond{\sY \sIntersection \sA} }[ \sA \in \symcal{A} ] .
	\end{equation*}
	Since \(\lthSigma\) is finite, the intersection is finite.
\end{proof}

Next, we generalize \(\lthSigma\)-sub-p-morphisms to the weakly transitive case.

\begin{definition}[\(\lthSigma\)-sub-p-morphism]
    We call a set \(\sX['] \sSubsetEq \sX\) a \(\lthSigma\)-\emph{subset} of \(\frmM\) \tdefiff{} \(\ElimCrit[t=\lthSigma, m=\frmM, u=\sX[']] = \sX \sMinus \sX[']\) and for all \(\lphi \in \lthSigma\), \(\vx \in \sX[']\) and \(\vy \in \relR{\vx} \sIntersection \relR[converse]{\vx}\) such that \(\vx \in \sem{\lphi} \iff \vy \in \sem{\lphi}\), there is \(\vy['] \in \fSpan[\sX[']]{\vx}\) with \(\vx \in \sem{\lphi} \iff \vy['] \in \sem{\lphi}\).
    A sub-p-morphism \(\ff\) from \(\frmM\) to \(\frmN\) is called a \(\lthSigma\)-\emph{sub-p-morphism} \tdefiff{} \(\ff[dom]\) is a \(\lthSigma\)-subset of \(\frmM\).
\end{definition}

If \(\frmM\) is transitive, then the latter requirement is void.
In this case, if \(\lthSigma\) is finite and, up to logical equivalence, closed under Boolean connectives, then we obtain Zakharyaschev's definition of \(\lthSigma\)-subreduction.
If \(\lthSigma\) is, up to logical equivalence, the set of formulas of modal depth up to \(\nn\), we obtain Fine's definition of \(\nn\)-subreduction.
As in Fine's and Zakharyaschev's definitions, it is immediate that every \(\lthSigma\)-subset is cofinal.
In addition, we obtain the following key property of \(\lthSigma\)-subsets.

\begin{lemma}\zlabel{lemma:main-property-Sigma-submodel}
	Let \(\sX[']\) be a \(\lthSigma\)-subset of \(\frmM\), \(\lphi \in \lthSigma\), \(\vx \in \sX\) and \(\vy \in \relR{\vx}\).
	Then there exists \(\vy['] \in \fSpan[\sX[']]{\vx}\) with \(\vy \in \sem{\lphi} \iff \vy['] \in \sem{\lphi}\).
\end{lemma}
\begin{proof}
    If \(\vy \in \sX[']\) then take \(\vy['] \coloneqq \vy\), and we are done.
	So suppose \(\vy \in \ElimCrit[m=\frmM, u=\sX[']]\).
	Then there is \(\vz \in \fSpan[\sX[']]{\vy}\) such that \(\vy \in \sem{\lphi} \iff \vz \in \sem{\lphi}\).
	If \(\vz \in \relR{\vx}\), then take \(\vy['] \coloneqq \vz\). Otherwise, \(\vz \notin \relR{\vx}\) and,
	by weak transitivity, \(\vx = \vz\), so \(\vy \in \relR{\vx} \sIntersection \relR[converse]{\vx}\) and \(\vy \in \sem{\lphi} \iff \vx \in \sem{\lphi}\).
	By the second property of \(\lthSigma\)-subsets, there is \(\vy['] \in \fSpan[\sX[']]{\vx}\) with \(\vx \in \sem{\lphi} \iff \vy['] \in \sem{\lphi}\).
\end{proof}

From this it immediately follows that every \(\lthSigma\)-subset of \(\frmM\) is strongly cofinal.
Moreover, a routine induction on formulas shows that they preserve truth of formulas in \(\lthSigma\).

\begin{proposition}\zlabel{prop:Sigma-sub-p-morphism-preserves-truth}
	Let \(\ff\) be a \(\lthSigma\)-sub-p-morphism from \(\frmM\) to \(\frmN\) and let \(\frmM[']\) be the restriction of \(\frmM\) to \(\ff[dom]\).
	Then, for every \(\vx \in \ff[dom]\),
	\begin{equation*}
		\fTy[\lthSigma]{\frmM, \vx} = \fTy[\lthSigma]{\frmM['], \vx} = \fTy[\lthSigma]{\frmN, \ff{\vx}} .
	\end{equation*}
\end{proposition}

From \zcref{lemma:main-property-Sigma-submodel,prop:Sigma-sub-p-morphism-preserves-truth} it also follows that, as in the transitive setting of Zakharyaschev, \(\lthSigma\)-sub-p-morphisms compose. 
However, we will not use this property.

In the selection construction we will be iteratively removing points that are eliminable w.r.t.\@ some upset \(\sU\), and extending this \(\sU\).
For the latter step, we introduce \(\sTopClust[\sU]\), the set of points in the \enquote{highest} clusters below \(\sU\).

\begin{definition}\zlabel{def:top-cluster}
    For an upset \(\sU\) in \(\frmM\) and \(\sD\coloneqq \sX\sMinus \sU\), we define
	\begin{equation*}
    	\sTopClust[\sU] \coloneqq\Set{\vx \in \sD}[
    	    \fSpan[\sD]{\vx} \sSubsetEq \relR[converse]{\vx}
    	]
        .
	\end{equation*}
\end{definition}

The next lemma follows straightforwardly from the above definition.

\begin{lemma}\zlabel{lemma:top-clust:upset}
	If \(\sU\) is an upset in \(\frmM\), then also \(\sTopClust[\sU] \sUnion \sU\) is an upset.
\end{lemma}

In the selection construction we will take quotients that are maximal w.r.t.\@ these upsets.
\begin{definition}
    For an upset \(\sU\) in \(\frmM\), we call \(\vx, \vy \in \sX\) \emph{bisimilar relative to \(\sU\)} if \(\vx\) and \(\vy\) are bismilar in the extension of \(\frmM\) with a fresh propositional variable that evaluates to \(\sU\).
    Equivalently, bisimilarity relative to \(\sU\) is the largest bisimulation equivalence on \(\frmM\) that does not identify points in \(\sU\) and \(\sX\sMinus\sU\).
    We say that \(\frmM\) is \emph{collapsed relative to \(\sU\)} if bisimilarity relative to \(\sU\) is the diagonal \(\relDiag[\sX] \coloneqq \Set{\homotuple{\vx, \vx}}[\vx \in \sX]\).
\end{definition}
 
    \section{Generalized Selection Construction}\zlabel{sec:selection-construction}
    In this section, we define the selection construction.
Compared to Fine's original construction, the definition of the top \(\sT[\odalpha]\) is modified to accommodate weakly transitive frames, quotienting out bisimilarity is postponed (and not even part of the construction), and the construction is generalized to infinitely many stages and infinite \(\lthSigma\), which we will use to re-derive and generalize Fine's Finite Width Theorem.

Let us fix a weakly transitive Kripke model \(\frmM = \structuple{\sX, \relR, \fVal}\).

\begin{construction}\zlabel{constr:selective-filtration}
We define a decreasing sequence of sets \(\sX[\odalpha] \sSubsetEq \sX\) and an increasing sequence of upsets \(\sU[\odalpha]\) in the restriction of \(\frmM\) to \(\sX[\odalpha]\) by transfinite induction as follows.

    For any ordinal \(\odalpha\) we write
    \begin{itemize}
        \item \(\frmM[\odalpha] = \structuple{\sX[\odalpha], \relR[\odalpha], \fVal[\odalpha]}\) for the restriction of \(\frmM\) to \(\sX[\odalpha]\),
        \item \(\relBisim[\odalpha]\) for bisimilarity relative to \(\sU[\odalpha]\) on \(\frmM[\odalpha]\),
\item \(\frmM[tilde, \odalpha] = \structuple{\sX[tilde, \odalpha], \relR[tilde, \odalpha], \fVal[tilde, \odalpha]}\) for the quotient of \(\frmM[\odalpha]\) through \(\relBisim[\odalpha]\), and
        \item \(\fq[\odalpha]\) for the natural quotient map from \(\frmM[\odalpha]\) to \(\frmM[tilde, \odalpha]\).
    \end{itemize}

    Define \(\sX[0] \coloneqq \sX\) and \(\sU[0] \coloneqq \sEmpty\).
    For the successor step, suppose \(\sX[\odalpha]\) and \(\sU[\odalpha]\) are already defined.
    Let \(\sT[\odalpha] \coloneqq \fq[\odalpha, preimage]{\sTopClust[\fq[\odalpha]{\sU[\odalpha]}]}\), where \(\sTopClust[\fq[\odalpha]{\sU[\odalpha]}]\) is taken in the model \(\frmM[tilde, \odalpha]\).
    Define
    \[
        \sU[\odalpha + 1] \coloneqq \sT[\odalpha] \sUnion \sU[\odalpha] \quad \text{and} \quad
        \sX[\odalpha + 1] \coloneqq \sX[\odalpha] \sMinus \ElimCrit[m=\frmM[\odalpha], u=\sU[\odalpha + 1]] .
    \]
    Notice that \(\sU[\odalpha] \sSubsetEq \sU[\odalpha + 1] \sSubsetEq \sX[\odalpha + 1] \sSubsetEq \sX[\odalpha]\) as required.
    Moreover, \(\sU[\odalpha + 1]\) is an upset in \(\frmM[\odalpha]\) by \zcref{lemma:top-clust:upset}.

    For the limit step, define
    \begin{equation*}
        \sU[\odlambda] \coloneqq \sLimitUnion[\odalpha < \odlambda] \sU[\odalpha]
        \quad \text{and} \quad
        \sX[\odlambda] \coloneqq \sX[\odlambda, '] \sMinus \ElimCrit[{m={\frmM[{restrict={\sX[\odlambda, ']}}]}, u={\sU[\odlambda]}}]
        \quad \text{where} \quad
        \sX[\odlambda, '] \coloneqq \sLimitIntersection[\odalpha < \odlambda] \sX[\odalpha]
        ,
    \end{equation*}
    and notice that \(\sU[\odlambda] \sSubsetEq \sX[\odlambda]\) is an upset in \(\frmM[\odlambda]\) as required.

    If \(\sX[\odalpha] = \sU[\odalpha]\) then we say that the construction  \emph{has terminated} at stage \(\odalpha\).
    We say that the construction \emph{terminates} \tdefiff{} it has terminated at stage \(\odalpha\) for some ordinal \(\odalpha\).
If \(\sX[\odalpha] \neq \sU[\odalpha]\) but there exists \(\odbeta < \odalpha\) such that \(\sU[\odalpha] = \sU[\odbeta]\) and \(\sX[\odalpha] = \sX[\odbeta]\), we say that the construction is \emph{stuck} at stage \(\odalpha\).
    Finally, if neither is the case, we say the construction \emph{progresses} at stage \(\odalpha\).
\end{construction}

Clearly, the construction either eventually terminates or eventually gets stuck.

\begin{lemma}\zlabel{prop:selective-filtration:termination}
	Let \(\odalpha\) and \(\odbeta\) be ordinals with \(\odalpha \leq \odbeta\).
	Then \begin{thmenumerate}
		\item\zlabel{prop:selective-filtration:termination:1}
		if the construction has terminated at stage \(\odalpha\), then it has terminated at stage \(\odbeta\),

		\item\zlabel{prop:selective-filtration:termination:2}
		if the construction is stuck at stage \(\odalpha\), then it is stuck at stage \(\odbeta\),

		\item\zlabel{prop:selective-filtration:termination:3}
		if the construction progresses at stage \(\odbeta\), then it progresses at stage \(\odalpha\), and

		\item\zlabel{prop:selective-filtration:termination:4}
		if \(\odalpha\) is a cardinal strictly larger than the cardinality of \(\sX\), then the construction has either terminated or is stuck at stage \(\odalpha\).
	\end{thmenumerate}
\end{lemma}

Therefore, we want a sufficient criterion for the construction not to get stuck.
We will give such a criterion in the next section.
First, we state some simple properties about the construction.

\begin{lemma}\zlabel{lemma:selective-filtration:basic-props}
    Let \(\odalpha\) be an ordinal.
    \begin{thmenumerate}
        \item\zlabel{lemma:head-is-union-tops}
        \(\sU[\odalpha] = \sSetUnion\Set{\sT[\odbeta]}[\odbeta < \odalpha]\) and \(\frmM[tilde,\odalpha]\) is collapsed relative to \(\fq[\odalpha]{\sU[\odalpha]}\).
        \item\zlabel{lemma:selective-filtration:head-bisim-static}
        If \(\odbeta\) is an ordinal with \(\odalpha \leq \odbeta\), and \(\vx, \vy \in \sU[\odalpha]\),
        then \(\vx \relBisim[\odalpha] \vy\) iff \(\vx \relBisim[\odbeta] \vy\).
        \item\zlabel{lemma:selective-filtration:bisim-next-head}
        If \(\vx, \vy \in \sU[\odalpha + 1]\) and \(\vx \relBisim[\odalpha] \vy\), then \(\vx \relBisim[\odalpha + 1] \vy\).
        \item\zlabel{lemma:selective-filtration:related-ty-equal-points-same-top}
        If \(\vx, \vy \in \sU[\odalpha]\)  such that \(\vy \in \relR{\vx}\) and \(\vx \relEquiv[ty=\lthSigma] \vy\), then there exists \(\odbeta < \odalpha\) such that \(\vx, \vy \in \sT[\odbeta]\).
    \end{thmenumerate}
\end{lemma}
\begin{proof}
    \zref{lemma:head-is-union-tops} is clear. For \zref{lemma:selective-filtration:head-bisim-static} note that both \(\relBisim[\odalpha]\) and \(\relBisim[\odbeta]\) restricted to \(\sU[\odalpha]\) are equal and for \zref{lemma:selective-filtration:bisim-next-head} additionally that \(\relBisim[\odalpha]\) does not identify points in \(\sU[\odalpha]\) with points in \(\sT[\odalpha]\).

	For \zref{lemma:selective-filtration:related-ty-equal-points-same-top}, note that by \zref{lemma:head-is-union-tops} there exist \(\odbeta, \odgamma < \odalpha\) such that \(\vx \in \sT[\odbeta]\) and \(\vy \in \sT[\odgamma]\).
	Since \(\sU[\odbeta + 1]\) is upward closed in \(\frmM[\odbeta]\) and \(\vx \in \sT[\odbeta] \sSubsetEq \sU[\odbeta + 1]\), we get \(\vy \in \sU[\odbeta + 1]\)  and hence \(\odgamma \leq \odbeta\).
	If the inequality were strict, then \(\vx \in \ElimCrit[m=\frmM[\odgamma], u=\sU[\odgamma + 1]]\), contradicting that \(\vx \in \sT[\odbeta] \sSubsetEq \sX[\odgamma + 1]\).
\end{proof}

\begin{lemma}\zlabel{lemma:selection-construction-sigma-subset}
	For any ordinal \(\odalpha\), the set \(\sX[\odalpha]\) is a \(\lthSigma\)-subset of \(\frmM\).
\end{lemma}
\begin{proof}
    For a successor ordinal \(\odalpha\), define \(\sX[\odalpha, '] \coloneqq \sX[\odalpha - 1]\), and for a limit ordinal \(\odlambda\), let \(\sX[\odlambda, ']\) be as in the construction.
    Write \(\frmM[\odbeta, ']\) for the restriction of \(\frmM\) to \(\sX[\odbeta, ']\).
    We prove by induction on ordinals \(\odalpha\) that \(\sX[\odalpha, ']\) is a \(\lthSigma\)-subset of \(\frmM\).
    Suppose, as induction hypothesis, that \(\sX[\odbeta, ']\) a \(\lthSigma\)-subset of \(\frmM\) for every \(\odbeta < \odalpha\).
    First, let \(\vx \in \sX \sMinus \sX[\odalpha, ']\), say \(\vx \in \ElimCrit[m=\frmM[{\odbeta, '}], u=\sU[\odbeta]]\) for some \(\odbeta < \odalpha\).
    By the induction hypothesis and \zcref{prop:Sigma-sub-p-morphism-preserves-truth}, \(\ElimCrit[m=\frmM[{\odbeta, '}], u=\sU[\odbeta]] \sSubsetEq \ElimCrit[m=\frmM, u=\sU[\odbeta]]\).
    Since \(\sU[\odbeta] \sSubsetEq \sX[\odalpha, ']\), we conclude that \(\vx \in \ElimCrit[m=\frmM, u=\sX[{\odalpha, '}]]\).

    Second, let \(\lphi \in \lthSigma\), \(\vx \in \sX[\odalpha, ']\) and \(\vy \in \relR{\vx} \sIntersection \relR[converse]{\vx}\) such that \(\vx \in \sem{\lphi}[\frmM] \iff \vy \in \sem{\lphi}[\frmM]\).
If \(\vy \in \sX[\odalpha, ']\) then we are done.
	Otherwise, \(\vy \in \ElimCrit[m=\frmM[{\odbeta, '}], u=\sU[\odbeta]]\) for some \(\odbeta < \odalpha\).
	Using the induction hypothesis and \zcref{prop:Sigma-sub-p-morphism-preserves-truth}, there exists \(\vy['] \in \sU[\odbeta] \sIntersection \relR{\vy}\) with \(\vy \in \sem{\lphi}[\frmM] \iff \vy['] \in \sem{\lphi}[\frmM]\).
	Suppose \(\vy['] \notin \relR{\vx}\).
	Then, by weak transitivity, \(\vy['] = \vx\), yielding \(\vx \in \sU[\odbeta]\).
	Since \(\sU[\odbeta]\) is an upset in \(\frmM[\odbeta, ']\), we get \(\vy \in \sU[\odbeta]\), but this contradicts \(\vy \in \ElimCrit[m=\frmM[{\odbeta, '}], u=\sU[\odbeta]]\).
	Therefore \(\vy['] \in \relR{\vx}\).
\end{proof}

The next two lemmata will be used in the proof of the Finite Width Theorem.

\begin{lemma}\zlabel{lemma:selective-filtration:head-conv-pre-wf}
    For any ordinal \(\odalpha\), the restriction of \(\frmM[\odalpha, tilde]\) to \(\fq[\odalpha]{\sU[\odalpha]}\) is conversely pre-well-founded. \end{lemma}
\begin{proof}
    By induction on \(\odalpha\).
    For \(\odalpha = 0\) this is trivial.
    The limit step is routine.

    For the successor step, suppose that the restriction of \(\frmM[\odalpha, tilde]\) to \(\sU \coloneqq \fq[\odalpha]{\sU[\odalpha]}\) is conversely pre-well-founded.
Then clearly also the restriction to \(\sU \sUnion \sTopClust[\sU]\) is conversely pre-well-founded.
    Note that, by construction,
    \(\fq[\odalpha]{\sU[\odalpha + 1]} = \sU \sUnion \sTopClust[\sU]\).
    However, by \zsubref{lemma:selective-filtration:basic-props}{lemma:selective-filtration:bisim-next-head}, \(\frmM[\odalpha + 1, tilde]\) restricted to \(\fq[\odalpha + 1]{\sU[\odalpha + 1]}\) is a p-morphic image of \(\frmM[\odalpha, tilde]\) restricted to \(\fq[\odalpha]{\sU[\odalpha + 1]}\).
    But the p-morphic image of a conversely pre-well-founded frame is conversely pre-well-founded.
\end{proof}

\begin{lemma}\zlabel{lemma:selective-filtration:head-no-eliminable-points}
    For any ordinal \(\odalpha\), no \(\vx[tilde] \in \fq[\odalpha]{\sU[\odalpha]}\) is \(\lthSigma\)-eliminable in \(\frmM[tilde, \odalpha]\).
\end{lemma}
\begin{proof}
    Let \(\vx \in \sU[\odalpha]\).
    Then there exists \(\odbeta < \odalpha\) such that \(\vx \in \sT[\odbeta]\).
    As \(\sU[\odbeta] \sUnion \sT[\odbeta]\) is an upset in \(\frmM[\odbeta]\), we get \(\relR[\odbeta]{\vx} \sSubsetEq \sU[\odbeta] \sUnion \sT[\odbeta]\).
     We show that \(\fq[\odalpha]{\vx}\) is not \(\lthSigma\)-eliminable.

    If for every \(\lphi \in \lthSigma\), there exists a \(\vy \in \relR{\vx}\sIntersection \sU[\odbeta]\) with \(\vx \in \sem{\lphi} \iff \vy \in \sem{\lphi}\), then \(\vx\) is \(\lthSigma\)-eliminable w.r.t.\@ \(\sU[\odbeta]\), a contradiction.
So there is a \(\lphi \in \lthSigma\) such that for every \(\vy \in \relR[\odbeta]{\vx}\) with \(\vx \in \sem{\lphi} \iff \vy \in \sem{\lphi}\), we have \(\vy \in \sT[\odbeta]\).
    Then \(\fq[\odbeta]{\vy} \in \fq[\odbeta]{\sT[\odbeta]} = \sTopClust[\fq[\odbeta]{\sU[\odbeta]}]\), yielding \(\fq[\odbeta]{\vx} \in \relR[\odbeta,tilde]{\fq[\odbeta]{\vy}}\).
    It follows from \zsubref{lemma:selective-filtration:basic-props}{lemma:selective-filtration:head-bisim-static} and \zref{lemma:selective-filtration:bisim-next-head} that \(\fq[\odalpha]{\vx} \in \relR[\odalpha, tilde]{\fq[\odalpha]{\vy}}\).
    We conclude that, since \(\fq[\odalpha]\) is a p-morphism, for all \(\vy[tilde] \in \relR[\odalpha,tilde]{\fq[\odalpha]{\vx}}\)
    with \(\fq[\odalpha]{\vx} \in \sem{\lphi} \iff \vy[tilde] \in \sem{\lphi}\), we have \(\fq[\odalpha]{\vx} \in \relR[\odalpha,tilde]{\vy[tilde]}\).
    Hence, \(\fq[\odalpha]{\vx}\) is not \(\lthSigma\)-eliminable.
\end{proof}
 
    \section{Covering via Stability and Bisimulation}\zlabel{sec:covering}
    In this section we prove that, in \zcref{constr:selective-filtration}, if the restriction of \(\frmM[\odalpha, tilde]\) to \(\fq[\odalpha]{\sU[\odalpha]}\) does not contain an infinite anti-chain, then \(\sT[\odalpha]\) covers \(\sX[\odalpha] \sMinus \sU[\odalpha]\).
In particular, it then follows that the construction is not stuck at stage \(\odalpha\).
To prove this result we first need to develop a theory of stable points, and construct a bisimulation equivalence on them.

We first introduce the general notion of stability w.r.t.\@ an equivalence relation or function.

\begin{definition}[Stable]
    Let \(\frF = \structuple{\sX, \relR}\) be a frame, \(\relEquiv\) an equivalence relation on \(\sX\), and \(\sY \sSubsetEq \sX\).
    Then a point \(\vy \in \sY\) is called \emph{\(\relEquiv\)-stable in} \(\sY\) iff
    \begin{equation*}
        \ForAll{\vy['] \in \relR[refl trans closure]{\vy} \sIntersection \sY}{\vy \relEquiv \vy[']}.
    \end{equation*}
    If \(\Function{\ff}{\sX}{\sZ}\) is a function, then a point \(\vy \in \sY\) is called \emph{\(\ff\)-stable in \(\sY\)} \tdefiff{} it is \(\relEquiv[\ff]\)-stable in \(\sY\), where \(\relEquiv[\ff]\) is the binary equivalence relation on \(\fDom{\ff}\) induced by \(\ff\), i.e.\@ \(\vy \relEquiv \vy[']\) iff \(\ff{\vy} = \ff{\vy[']}\).
\end{definition}

Now we define the set of stable points we are interested in, and an equivalence relation on it that will turn out to be a bisimulation equivalence.
This set and bisimulation lie at the heart of both the original version of Fine's selective filtration theorem and Fine's finite width theorem.

Fix a weakly transitive Kripke model \(\frmM = \structuple{\sX, \relR, \fVal}\) with an upset \(\sU \sSubsetEq \sX\) and \(\sD \coloneqq \sX \sMinus \sU\).

\begin{construction}\zlabel{constr:stable-span-ty-equiv}
We define
	\begin{equation*}
		\sStab \coloneqq \Set{\vx \in \sD}[\text{\(\vx\) is \(\fSpan[\sU]\)- and \(\fComp{\fTy[\lthProp] \circ \fSpan[\sD]}\)-stable in \(\sD\)}].
	\end{equation*}
	Let \(\relEquiv[\mathrm{tysp}]\) and \(\relEquiv[\sU]\) be the equivalence relations induced by \(\fTy[\lthProp] \circ \fSpan[\sD]\) and the characteristic function of \(\sU\), respectively, and
	define \(\relEquiv[s]\) to be the intersection of these two equivalence relations with \(\relEquiv[ty=\lthProp]\) and \(\relEquiv[span=\sU]\), restricted to \(\sStab \sUnion \sU\). Note that \(\sStab \sUnion \sU\) is an upset.
\end{construction}

The set \(\sStab\) consists of points in \(\sD\) such that while going up in the accessibility relation and staying in \(\sD\) neither the \(\sU\)-span nor the set of \(\lthProp\)-types seen in \(\sD\) change.

\begin{lemma}\zlabel{lemma:stable-span-ty-equiv-bisim}
	The equivalence relation \(\relEquiv[s]\) is a bisimulation equivalence on \(\sStab \sUnion \sU\).
\end{lemma}
\begin{proof}
    Since \({\relEquiv[s]} \sSubsetEq {\relEquiv[ty=\lthProp]}\), it preserves atomic propositions.
	For the back- and forth-conditions, let \(\vx \relEquiv[s] \vy\) and \(\vx['] \in \relR{\vx}\).
		Suppose \(\vx['] \in \sU\).
		Then \(\vx['] \in \fSpan[\sU]{\vx} = \fSpan[\sU]{\vy}\), hence \(\vx['] \in \relR{\vy}\).
		Note that \(\vx['] \relEquiv[s] \vx[']\), by reflexivity, since \(\vx['] \in \sU \sSubsetEq \sStab \sUnion \sU\).

		Otherwise, \(\vx['] \in \sD\) and, since \(\sD\) is a downset, also \(\vx \in \sD\).
		Hence \(\vx \in \sStab\) and therefore \(\vy \in \sStab\) as well.
		Then \(\fTy[\lthProp]{\vx[']} \in \fComp{\fTy[\lthProp] \circ \fSpan[\sD]}{\vx} = \fComp{\fTy[\lthProp] \circ \fSpan[\sD]}{\vy}\), so there exists \(\vy['] \in \relR{\vy} \sIntersection \sD\) with \(\vx['] \relEquiv[ty=\lthProp] \vy[']\).
By the \(\fSpan[\sU]\)-stability of \(\vx\), we get \(\vx \relEquiv[span=\sU] \vx[']\)
		and analogously \(\vy \relEquiv[span=\sU] \vy[']\), so
		we conclude that
		\(
		    \vx['] \relEquiv[span=\sU] \vx \relEquiv[span=\sU] \vy \relEquiv[span=\sU] \vy[']
		\)
		as desired.
		The argument for \(\relEquiv[\mathrm{tysp}]\) is analogous.
		Finally, since \(\vx['], \vy['] \in \sD\) we obtain \(\vx['] \relEquiv[\sU] \vy[']\).
\end{proof}

For technical reasons we will work with \(\fSpan[\sD]{\sStab}\) instead of \(\sStab\).
This effectively removes from \(\sStab\) the irreflexive points \enquote{at the bottom}. The following lemma, which easily follows from the fact that \(\sStab \sUnion \sU\) is an upset, shows that this is not a problem.

\begin{lemma}
	\(\fSpan[\sD]{\sStab} \sSubsetEq \sStab\), and \(\fSpan[\sD]{\sStab} \sUnion \sU\) is an upset.
\end{lemma}

We use \zcref{constr:stable-span-ty-equiv}, in particular the set \(\fSpan[\sD]{\sStab}\) and the restriction of \(\relEquiv[s]\) to it, as a tool for reasoning about \(\sTopClust{\sU}\).

\begin{lemma}\zlabel{lemma:stable-span-ty-top-cluster-covers-D}
    If \(\lthProp\) is finite, and the restriction of \(\frmM\) to \(\sU\) has the finite cover property,
	then \(\fSpan[\sD]{\sStab} \sUnion \sTopClust[\sU]\) covers \(\sD\).
\end{lemma}
\begin{proof}
    Let \(\sT \coloneqq \fSpan[\sD]{\sStab} \sUnion \sTopClust[\sU]\) and \(\vx \in \sD\sMinus \sT\).
    Suppose that \(\relR{\vx}\sIntersection \sTopClust[\sU] = \emptyset\). Then each \(\vy \in \fSpan[\sD]{\vx}\) has a strict successor in \(\sD\). We will first show that \(\vx\) has a \(\fSpan[\sU]\)-stable successor in \(\sD\).

    Suppose for a contradiction that \(\vx\) has no \(\fSpan[\sU]\)-stable successor in \(\sD\).
    Then, since, by weak-transitivity, \(\fSpan[\sU]\) is monotone decreasing with respect to \(\relR\),  using the axiom of dependent choice, we can construct an injective monotone sequence \(\Function{\ff}{\structuple{\Naturals, \leq}}{\structuple{\sD, {\relR} \sIntersection \sD[pow=2]}}\) such that \(\fSpan[\sU]\circ \ff\) is injective and for each \(\nii \in \Naturals\), \(\fSpan[\sU]{\ff{\nii+1}} \subseteq \fSpan[\sU]{\ff{\nii}}\).
    Hence, since \(\frmM\) restricted to \(\sU\) has the finite cover property, there exist \(\nii,\nj\in \Naturals\) with \(\nii < \nj\) such that \(\fSpan[\sU]{\ff{\nii}} = \fSpan[\sU]{\ff{\nj}}\), a contradiction.

    Therefore, by weak transitivity, we may assume that \(\vx\) is \(\fSpan[\sU]\)-stable.
    It remains to show that \(\vx\) has a \(\fComp{\fTy[\lthProp] \circ \fSpan[\sD]}\)-stable successor which will, in particular, be \(\fSpan[\sU]\)-stable.
    Since \(\lthProp\) is finite, \(\fTy[\lthProp]\) attains only finitely many values.
    Moreover, by weak transitivity, the function \(\fSpan[\sD]\)
    is strictly decreasing on \(\sD\) with respect to the strict successor relation.
    Hence, there exists a \(\fComp{\fTy[\lthProp] \circ \fSpan[\sD]}\)-stable \(\vy['] \in \fSpan[\sD]{\vx}\).
    Now there exists a strict successor \(\vy \in \sD\) of \(\vy[']\) which is also \(\fComp{\fTy[\lthProp] \circ \fSpan[\sD]}\)-stable, and, by weak transitivity, \(\vy \in \fSpan[\sD]{\vx}\).
\end{proof}

\begin{lemma}\zlabel{lemma:stable-span-ty-in-top-cluster}
    If \(\frmM\) is collapsed relative to \(\sU\),
then \(\fSpan[\sD]{\sStab} \sSubsetEq \sTopClust{\sU}\).
 \end{lemma}
\begin{proof}
     Let \(\vx['] \in \fSpan[\sD]{\sStab}\) and \(\vy \in \fSpan[\sD]{\vx[']}\).
     We prove that \(\vy \in \relR[converse]{\vx[']}\), i.e.\@ \(\relR{\vy, \vx[']}\).
     Pick \(\vx \in \sStab\) such that \(\vx['] \in \fSpan[\sD]{\vx}\).
     Then \(\vy \in \relR[refl trans closure]{\vx}\), so by stability \(\vx \relEquiv[\mathrm{tysp}] \vy\).
     Therefore, as \(\vx['] \in \fSpan[\sD]{\vx}\), there exists \(\vy['] \in \fSpan[\sD]{\vy}\) with \(\vx['] \relEquiv[ty=\lthProp] \vy[']\).
     Now \(\vx['] \relEquiv[ty=\lthProp] \vy[']\), and by stability \(\vx['] \relEquiv[span=\sU] \vy[']\) and \(\vx['] \relEquiv[\mathrm{tysp}] \vy[']\).
     Since \(\vx[']\) and \(\vy[']\) are both in \(\sD\), \(\vx['] \relEquiv[\sU] \vy[']\).
     We conclude that \(\vx['] \relEquiv[s] \vy[']\).

     By \zcref{lemma:stable-span-ty-equiv-bisim},    \(\relEquiv[s]\) is a bisimulation equivalence on the upset \(\sStab \sUnion \sU\).
     So it extends to a bisimulation equivalence on \(\sX\) that, in particular, does not identify points in \(\sU\) and \(\sD\).
     Hence \(\vx[']\) and \(\vy[']\) are bisimilar relative to \(\sU\) and,
     since \(\frmM\) is collapsed relative to \(\sU\), it follows that \(\vx['] = \vy[']\).
     But we know that \(\relR{\vy, \vy[']}\), hence \(\relR{\vy, \vx[']}\), and we are done.
\end{proof}

These two lemmata combine into the sought criterion for \zcref{constr:selective-filtration} to not get stuck.

\begin{proposition}\zlabel{lemma:selective-filtration:finite-width-top-implies-covering}
    Suppose that \(\lthProp\) is finite and let \(\odalpha\) be an ordinal.
    Moreover, suppose that \(\frmM[\odalpha, tilde]\) restricted to \(\fq[\odalpha]{\sU[\odalpha]}\) in \zcref{constr:selective-filtration} is of finite width.
    Then \(\sT[\odalpha]\) covers \(\sX[\odalpha] \sMinus \sU[\odalpha]\), and, in particular, the construction does not get stuck at stage \(\odalpha\).
\end{proposition}
\begin{proof}
    We apply \zcref{constr:stable-span-ty-equiv} to \(\frmM[\odalpha, tilde]\) with upset \(\sU \coloneqq \fq[\odalpha]{\sU[\odalpha]}\), defining \(\sStab\).

    By assumption \(\frmM[\odalpha, tilde]\) restricted to \(\sU\) is of finite width, and by \zcref{lemma:selective-filtration:head-conv-pre-wf} it is conversely pre-well-founded.
    Therefore, it has the finite cover property and,
    by \zcref{lemma:stable-span-ty-top-cluster-covers-D}, \(\fSpan[\sD]{\sStab} \sUnion \sTopClust[\sU]\) covers \(\sD \coloneqq \sX[\odalpha, tilde] \sMinus \sU\).
    Moreover, \(\frmM[\odalpha, tilde]\) is collapsed relative to \(\sU\).
    Hence, by \zcref{lemma:stable-span-ty-in-top-cluster}, \(\fSpan[\sD]{\sStab} \sSubsetEq \sTopClust[\sU]\), so \(\sTopClust[\sU]\) covers \(\sD\).
    Therefore, \(\sT[\odalpha] = \fq[\odalpha, preimage]{\sTopClust[\sU]}\) covers \(\fq[\odalpha, preimage]{\sD} = \fq[\odalpha, preimage]{\sX[\odalpha, tilde] \sMinus \fq[\odalpha]{\sU[\odalpha]}} = \sX[\odalpha] \sMinus \sU[\odalpha]\), where we use that \(\relBisim[\odalpha]\) by definition does not identify points in \(\sU[\odalpha]\) with points outside of it.
\end{proof}

Together with \zcref{prop:selective-filtration:termination,lemma:selective-filtration:head-conv-pre-wf}, we obtain the following result.

\begin{corollary}\zlabel{cor:finite-width-termination}
If, in \zcref{constr:selective-filtration}, \(\lthProp\) is finite and \(\frmM\) is of finite width, then the construction has terminated at some stage \(\odalpha\) and \(\frmM[tilde, \odalpha]\) has the finite cover property.
\end{corollary}
 
    \section{Definability via the Diagram Bisimulation}\zlabel{sec:diagram-constr}
    For the selection theorem, we need one more technical result: for finite \(\lthSigma\), the sets \(\sU[\nn]\) and \(\sX[\nn]\) from \zcref{constr:selective-filtration} are definable for each \(\nn \in\Naturals\).
To obtain this result, we establish a criterion which ensures that \(\sTopClust\) is definable and introduce another bisimulation equivalence on \(\sTopClust\), which we will call the \emph{diagram bisimulation}.
This relation is defined in three stages, where each next stage is defined in terms of the previous stages  by conditions resembling the back- and forth-condition for bisimulations.

We fix a weakly transitive Kripke model \(\frmM = \structuple{\sX, \relR, \fVal}\) with upset  \(\sU \sSubsetEq \sX\) and \(\sD \coloneqq \sX \sMinus \sU\).

\begin{construction}\zlabel{constr:diagram-equiv:equiv-01}

    Let \(\relEquiv[0]\) be the intersection \({\relEquiv[ty=\lthProp]} \sIntersection {\relEquiv[span=\sU]}\) restricted to \(\sD\) and define the relation \(\relPreEquiv[1]\) on \(\sD\) by
    \(\vx \relPreEquiv[1] \vx[']\) if and only if the following two conditions hold:
    \begin{itemize}
        \item[(C1)]
        for every \(\vy \in \fSpan[\sD]{\vx}\), there exists \(\vy['] \in \fSpan[\sD]{\vx[']}\) such that \(\vy \relEquiv[0] \vy[']\), and
        \item[(C2)]
        for all \(\vy \in \fSpan[\sD]{\vx}\), \(\vz \in \fSpan[\sD]{\vy}\), \(\vy['] \in \fSpan[\sD]{\vx[']}\) with \(\vy['] \relEquiv[0] \vy \relEquiv[0] \vz\), there exists \(\vz['] \in \fSpan[\sD]{\vy[']}\) with \(\vz \relEquiv[0] \vz[']\).
\end{itemize}
    The conditions are depicted in \zcref{fig:diagram-equiv:diagram-preord-1}.
    We define \(\mathord{\relEquiv[1]}\) to be the intersection of \(\relPreEquiv[1]\) and its converse, and write \(\relEquiv[01]\) for the intersection \(\mathord{\relEquiv[0]} \sIntersection \mathord{\relEquiv[1]}\). Note that, in general, only \({\relEquiv[1,dom]} \sSubsetEq {\relPreEquiv[1,dom]}\).
\end{construction}

\begin{figure}
    \centering
\begin{tikzcd}[ampersand replacement=\&]
    	y \& {\exists y'} \\
    	x \& {x'}
    	\arrow["{\relEquiv[0]}", dashed, no head, from=1-1, to=1-2]
    	\arrow[from=2-1, to=1-1]
    	\arrow["{\relPreEquiv[1]}"', from=2-1, to=2-2]
    	\arrow[dashed, from=2-2, to=1-2]
    \end{tikzcd}
    and
\begin{tikzcd}[ampersand replacement=\&]
    	z \& {\exists z'} \\
    	y \& {y'} \\
    	x \& {x'}
    	\arrow["{\relEquiv[0]}", dashed, no head, from=1-1, to=1-2]
    	\arrow["{\relEquiv[0]}"', shift right=3, no head, from=1-1, to=2-1]
    	\arrow[from=2-1, to=1-1]
    	\arrow["{\relEquiv[0]}", no head, from=2-1, to=2-2]
    	\arrow[dashed, from=2-2, to=1-2]
    	\arrow[from=3-1, to=2-1]
    	\arrow["{\relPreEquiv[1]}"', from=3-1, to=3-2]
    	\arrow[from=3-2, to=2-2]
    \end{tikzcd}
    \caption{Diagram describing the definition of \(\relPreEquiv[1]\) in \zcref{constr:diagram-equiv:equiv-01}. All variables live in \(\sD\).}\zlabel{fig:diagram-equiv:diagram-preord-1}
\end{figure}

The next lemma easily follows from the transitivity of \(\relEquiv[0]\) and the definition of \(\relPreEquiv[1]\).

\begin{lemma}
    The relation \(\relPreEquiv[1]\) is transitive, and \(\relEquiv[1]\) is an equivalence relation on its domain.
\end{lemma}

The next lemma shows that the domain of \(\relEquiv[1]\) contains \(\sTopClust{\sU}\). 

\begin{lemma}\zlabel{lemma:dom-equiv-1-top}
    \(\sTopClust{\sU} \sSubsetEq {\relEquiv[1,dom]}\).
\end{lemma}
\begin{proof}
    Let \(\vx \in \sTopClust[\sU]\).
    Note that it suffices to show \(\vx \relPreEquiv[1] \vx\).
    We check \conda{} and \condb{} for \(\relPreEquiv[1]\).
    For \conda{} note that if \(\vy \in \fSpan[\sD]{\vx}\),
    then, since \(\relEquiv[0]\) is reflexive, we have  \(\vy \relEquiv[0] \vy\).

    For \condb{} assume that  \(\vy, \vy['] \in \fSpan[\sD]{\vx}\) and \(\vz \in \fSpan[\sD]{\vy}\) with \(\vy['] \relEquiv[0] \vy \relEquiv[0] \vz\).
    There are two cases.
    If \(\vy = \vy[']\), then we are done.
    Otherwise, \(\vy \neq \vy[']\) and, since \(\vx \in \sTopClust[\sU]\), \(\vx \in \relR{\vy[']}\).
    By weak transitivity, \(\vy \in \relR{\vy[']}\), and by assumption \(\vz \relEquiv[0] \vy\).
\end{proof}

As the aim of this section is to prove a definability result for \(\sTopClust{\sU}\), we are interested in the definability of the equivalence classes of the introduced relations.

\begin{lemma}\zlabel{lemma:span-equiv-definable}
    Let \(\frmM = \structuple{\sX, \relR, \fVal}\) be a model, \(\vx \in \sX\), and let \(\sY \sSubsetEq \sX\) be such that \(\sY\) is finite and every point in \(\sY\) is definable.
    Then \(\sEquivClass[span=\sY]{\vx}\) is definable in \(\frmM\).
\end{lemma}
\begin{proof}
	Notice that
	\(
		\sEquivClass[span=\sY]{\vx}
		=
		\lBox{\sY \limplies \fSpan[\sY]{\vx}}
		\sIntersection
		\sSetIntersection\Set{\lDiamond\Set{\vy}}[\vy \in \fSpan[\sY]{\vx}]
		,
	\)
	which is definable, since the intersection is finite and \(\sY\), \(\fSpan[\sY]{\vx}\), and each \(\Set{\vy}\) are definable.
\end{proof}

The next lemma is the first ingredient to our definability result. 

\begin{lemma}\zlabel{lemma:diagram-equiv:equiv-01:finite-definable-equiv-classes}
Suppose that \(\sU\) and \(\lthProp\) are finite and $\sI \coloneqq {\relEquiv[1,dom]}$.
\begin{thmenumerate}
	    \item\zlabel{lemma:diagram-equiv:equiv-01:finite-definable-equiv-classes:1}
        \(\sD[quotient={\relEquiv[0]}]\) is finite, and if every point in \(\sU\) is definable, then each set in \(\sD[quotient={\relEquiv[0]}]\) is definable.
\item\zlabel{lemma:diagram-equiv:equiv-01:finite-definable-equiv-classes:3}
        \(\sI[quotient={\relEquiv[1]}]\) is finite, and if every point in \(\sU\) is definable, then each set in \(\sI[quotient={\relEquiv[1]}]\) is definable.
	\end{thmenumerate}
\end{lemma}
\begin{proof}
    \zref{lemma:diagram-equiv:equiv-01:finite-definable-equiv-classes:1}
Note that both \(\fTy[\lthProp]\) and \(\fSpan[\sU]\) attain only finitely many values on \(\sD\), so \(\sD[quotient={\relEquiv[0]}]\) is finite.
        For the second part suppose that every point in \(\sU\) is definable. The set \(\lthProp\) is finite, so the \(\relEquiv[ty=\lthProp]\)-equivalence class of \(\vx\) is definable, since we have
        \begin{equation*}
		\sEquivClass[ty=\lthProp]{\vx}
		= \sSetIntersection\paren{
    		\Set{
    		    \sem{\lap}
    		}[\lap \in  \fTy[\lthProp]{\vx}]
    		\sUnion
    		\Set{
    		    \sem{\lnot\lap}
    		}[\lap \in \lthProp\sMinus \fTy[\lthProp]{\vx}]
		}.
	    \end{equation*}
        Moreover, since \(\sU\) is finite and each point in it is definable, by \zcref{lemma:span-equiv-definable}, \(\sEquivClass[span=\sY]{\vx}\) is definable. Hence, \(\sEquivClass[0]{\vx} =  \sEquivClass[ty=\lthProp]{\vx} \sIntersection \sEquivClass[span=\sY]{\vx} \sIntersection \sD\) is definable.

		\zref{lemma:diagram-equiv:equiv-01:finite-definable-equiv-classes:3} We will show that there are only finitely many \(\relPreEquiv[1]\)-upsets which under the additional assumption are all definable. Then, the same holds for \(\relPreEquiv[1]\)-downsets, since every downset is the complement of an upset, and it follows that there are only finitely many  \(\relEquiv[1]\)-equivalence classes which under the additional assumption are all definable, since each such equivalence class is the intersection of a \(\relPreEquiv[1]\)-upset with a \(\relPreEquiv[1]\)-downset.
		For \(\vx \in \sD\) we define the two sets
        \begin{align*}
            \sY[\vx]
            &\coloneqq
            \sD
            \sIntersection
            \sSetIntersection\Set{
                \lDiamond\sEquivClass[0]{\vy}
            }[\vy \in \fSpan[\sD]{\vx}]
            ,
            \\
            \sY[', \vx]
            &\coloneqq
            \sD
            \sIntersection
            \sSetIntersection\Set{
                \lBox{\sEquivClass[0]{\vy} \limplies \lDiamond\sEquivClass[0]{\vz}}
            }[
                \vy \in \fSpan[\sD]{\vx},
                \vz \in \fSpan[\sD]{\vy},
                \vy \relEquiv[0] \vz
            ]
            .
        \end{align*}
        It is straightforward to check that \(\sY[\vx]\) is the set of points \(\vx[']\) such that condition (C1) of \(\vx \relPreEquiv[1] \vx[']\) is satisfied, and \(\sY[', \vx]\) is the set of points \(\vx[']\) such that condition (C2) of \(\vx \relPreEquiv[1] \vx[']\) is satisfied.
        Hence \(\sY[\vx] \sIntersection \sY[', \vx]\) is the \(\relPreEquiv[1]\)-upset generated by \(\vx\).
        It follows, by \zref{lemma:diagram-equiv:equiv-01:finite-definable-equiv-classes:1}, that there are only finitely many \(\relEquiv[0]\)-equivalence classes, so \(\sY[\vx]\) and \(\sY[', \vx]\) attain only finitely many values for \(\vx \in \sD\).
        Therefore, there are finitely many point-generated \(\relPreEquiv[1]\)-upsets and as \(\relPreEquiv[1]\)-upset is a union of point-generated upsets, there are only finitely many \(\relPreEquiv[1]\)-upset.

        Finally, if each point in \(\sU\) is definable, then, by \zref{lemma:diagram-equiv:equiv-01:finite-definable-equiv-classes:1}, it follows straightforwardly that for each \(\vx \in \sD\) the sets \(\sY[\vx]\) and \(\sY[', \vx]\) are definable, yielding that every \(\relPreEquiv[1]\)-upset is definable.
\end{proof}

Using the relations \(\relEquiv[0]\)  and \(\relEquiv[01]\) we can construct a full bisimulation.\smallskip

\noindent\begin{minipage}{0.8\textwidth}
\begin{construction}[extending \zcref{constr:diagram-equiv:equiv-01}]\zlabel{constr:diagram-equiv:equiv-2}
    Let \(\relEquiv[d]\) be the largest relation contained in \(\relEquiv[01]\) such that for all \(\vx \relEquiv[d] \vx[']\), \(\vy \in \fSpan[refl trans closure, \sD]{\vx}\) and \(\vy['] \in \fSpan[refl trans closure, \sD]{\vx[']}\), if \(\vy \relEquiv[0] \vy[']\) then \(\vy \relEquiv[1] \vy[']\).
    This definition is diagrammatically depicted on the right, where arrows with \(\ast\) at the tail denote edges in \(\relR[refl trans closure]\).

\end{construction}
\end{minipage}
\begin{minipage}{0.2\textwidth}
\centering
\begin{tikzcd}
	   \vy & {y\smash{^\prime}\!\!} \\
	   \vx & {x\smash{^\prime}\!\!}
	   \arrow["{\relEquiv[0]}"', shift right, no head, from=1-1, to=1-2]
	   \arrow["{\relEquiv[1]}", shift left, dashed, no head, from=1-1, to=1-2]
	   \arrow[from=2-1, to=1-1,ast->]
	   \arrow["{\relEquiv[01]}", no head, from=2-1, to=2-2]
	   \arrow[from=2-2, to=1-2,ast->]
    \end{tikzcd}
\end{minipage}

\begin{lemma}\zlabel{lemma:diagram-equiv:equiv-2:bisim}
    \begin{thmenumerate}
        \item\zlabel{lemma:diagram-equiv:equiv-2:upward-0-upgrade}
        Let \(\vx, \vx['], \in \sD\), \(\vy \in \fSpan[\sD]{\vx}\) and \(\vy['] \in \fSpan[\sD]{\vx[']}\) with \(\vx \relEquiv[d] \vx[']\) and \(\vy \relEquiv[0] \vy[']\).
	    Then \(\vy \relEquiv[d] \vy[']\).
	    In particular, \(\sU \sUnion {\relEquiv[d,dom]}\) is an upset.
        \item\zlabel{lemma:diagram-equiv:equiv-2:bisim:item}
        The relation \({\relEquiv[d]} \sUnion {\relDiag[\sU]}\) is a bisimulation equivalence on \(\sU \sUnion {\relEquiv[d,dom]}\).
    \end{thmenumerate}
\end{lemma}
\begin{proof}
    \zref{lemma:diagram-equiv:equiv-2:upward-0-upgrade}
    By the defining property of \(\relEquiv[d]\), we see that \(\vy \relEquiv[1] \vy[']\), hence \(\vy \relEquiv[01] \vy[']\).
	Suppose \(\vz \in \fSpan[refl trans closure, \sD]{\vy}\) and \(\vz['] \in \fSpan[refl trans closure, \sD]{\vy[']}\) with \(\vz \relEquiv[0] \vz[']\).
	By weak transitivity, \(\vz \in \fSpan[refl trans closure, \sD]{\vx}\) and \(\vz['] \in \fSpan[refl trans closure, \sD]{\vx[']}\).
	Then, since \(\vx \relEquiv[d] \vx[']\), it follows that \(\vz \relEquiv[1] \vz[']\).
	We conclude that \(\vy \relEquiv[d] \vy[']\).

    \zref{lemma:diagram-equiv:equiv-2:bisim:item}
    Clearly \({\relEquiv[d]} \sUnion {\relDiag[\sU]}\) is an equivalence relation.
    For the preservation of atomic propositions, notice that \({\relEquiv[d]} \sSubsetEq {\relEquiv[0]} \sSubsetEq {\relEquiv[ty=\lthProp]}\).
	For the back- and forth-conditions, suppose \(\vx \relEquiv[d] \vx[']\) and \(\vy \in \relR{\vx}\).
	If \(\vx \in \sU\) or \(\vx['] \in \sU\) then \(\vx = \vx[']\), and hence \(\vy \in \relR{\vx[']}\).
	Otherwise, \(\vx,\vx['] \in {\relEquiv[d,dom]}\).
    If \(\vy \in \sU\), then \(\vy \relEquiv[d] \vy\)  and \(\vy \in \fSpan[\sU]{\vx} = \fSpan[\sU]{\vx[']}\), since \(\vx \relEquiv[span=\sU] \vx[']\).
So suppose \(\vy \in \sD\).
	Since \(\vx \relEquiv[1] \vx[']\), by condition (C1) of \(\relEquiv[1]\), there exists \(\vy['] \in \fSpan[\sD]{\vx[']}\) such that \(\vy \relEquiv[0] \vy[']\) and \zref{lemma:diagram-equiv:equiv-2:upward-0-upgrade} yields \(\vy \relEquiv[d] \vy[']\).
\end{proof}

The usefulness of this bisimulation equivalence \(\relEquiv[d]\) lies in the fact that its equivalence classes are, under mild conditions on \(\sU\), definable and finitely many, like those of \(\relEquiv[0]\) and \(\relEquiv[1]\).

\begin{lemma}\zlabel{lemma:diagram-equiv:finite-definable-equiv-classes}
Suppose that \(\sU\) and \(\lthProp\) are finite. Then there are only finitely many \(\relEquiv[d]\)-equivalence classes. If, in addition, each point in \(\sU\) is definable, then each \(\relEquiv[d]\)-equivalence class is definable and, in particular, \({\relEquiv[d,dom]}\) is definable.
\end{lemma}
\begin{proof}
    Notice that for \(\vx \in {\relEquiv[d,dom]}\), we have
    \(
        {\sEquivClass[d]{\vx}}
        =
        {\sEquivClass[01]{\vx}} \sIntersection
        \sSetIntersection\Set{
            \lBox[refl closure]{
                {\sEquivClass[0]{\vy}} \limplies {\sEquivClass[1]{\vy}}
            }
        }[\vy \in \Set{\vx} \sUnion \fSpan[refl trans closure, \sD]{\vx}]
    \).
    By \zcref{lemma:diagram-equiv:equiv-01:finite-definable-equiv-classes}, all the involved sets admit only finitely many values, and under the additional assumption they are definable.
\end{proof}

The next lemma shows that \(\sTopClust[\sU]\) is contained in \(\relEquiv[d,dom]\). 

\begin{lemma}\zlabel{lemma:diagram-equiv:tiptop-contains-topclust}
\(\sTopClust[\sU] \sSubsetEq {\relEquiv[d,dom]}\).
\end{lemma}

\begin{proof}
    Let \(\vx \in \sTopClust[\sU]\). By \zcref{lemma:dom-equiv-1-top}, we have \(\sTopClust{\sU} \sSubsetEq {\relEquiv[1,dom]}\).

    Hence, for proving that \(\vx \in {\relEquiv[d,dom]}\), we need to show that for all \(\vy, \vy['] \in \fSpan[refl trans closure, \sD]{\vx}\), if \(\vy \relEquiv[0] \vy[']\) then \(\vy \relEquiv[1] \vy[']\).
    So let \(\vy, \vy['] \in \fSpan[refl trans closure, \sD]{\vx}\) with \(\vy \relEquiv[0] \vy[']\).
    It suffices to show \(\vy \relPreEquiv[1] \vy[']\), since \(\relEquiv[0]\) is symmetric.
    So we need to check the two defining properties \conda{} and \condb{}.

    For \conda{}, note first that, by definition of \(\sTopClust[\sU]\), we have \(\vx \in \fSpan[refl trans closure,\sD]{\vy[']}\) and \(\vy,\vy[']\in\sTopClust[\sU]\).
    Let \(\vz \in \relR{\vy}\).
    We show that there exists a \(\vz['] \in \fSpan[\sD]{\vy[']}\) with \(\vz \relEquiv[0] \vz[']\). We distinguish two cases.
    \begin{proofcases}
        \item
        Suppose \(\relR{\vx, \vz}\).
        This case is depicted diagrammatically in \zcref{fig:diagram-equiv:proof-tiptop-contains-topclust-condition-1-case-1}.

    \begin{figure}
    \centering
    \begin{tabular}{c|ccc}
\begin{tikzcd}
           	z &  \\
           	y & {y'} \\
           	x & x
           	\arrow[from=2-1, to=1-1]
           	\arrow["{\equiv_0}", no head, from=2-1, to=2-2]
           	\arrow[shift left, from=2-2, to=3-2,ast->]
           	\arrow[curve={height=-12pt}, from=3-1, to=1-1]
           	\arrow[from=3-1, to=2-1,ast->]
           	\arrow["{\equiv_1}", no head, from=3-1, to=3-2]
           	\arrow[shift left, from=3-2, to=2-2,ast->]
        \end{tikzcd}
        &
\begin{tikzcd}
           	z & {z} \\
           	y & {y'} \\
           	x & x
           	\arrow["{\equiv_0}",  no head, from=1-1, to=1-2]
           	\arrow[from=2-1, to=1-1]
           	\arrow["{\equiv_0}", no head, from=2-1, to=2-2]
           	\arrow[shift left, from=2-2, to=3-2,ast->]
            \arrow[dashed, from=2-2, to=1-2]
           	\arrow[curve={height=-12pt}, from=3-1, to=1-1]
           	\arrow[from=3-1, to=2-1,ast->]
           	\arrow["{\equiv_1}", no head, from=3-1, to=3-2]
           	\arrow[shift left, from=3-2, to=2-2,ast->]
        \end{tikzcd}
        &
\begin{tikzcd}
           	{y'} & {z'} \\
           	y & {y'} \\
           	x & x
           	\arrow["{\equiv_0}", dashed, no head, from=1-1, to=1-2]
           	\arrow["{\equiv_0}", shift left, no head, from=2-1, to=1-1]
           	\arrow[shift right, from=2-1, to=1-1]
           	\arrow["{\equiv_0}", no head, from=2-1, to=2-2]
           	\arrow[dashed, from=2-2, to=1-2]
           	\arrow[from=3-1, to=2-1]
           	\arrow["{\equiv_1}", no head, from=3-1, to=3-2]
           	\arrow[from=3-2, to=2-2]
        \end{tikzcd}
        &
\begin{tikzcd}
           	{y'} & x \\
           	x & {y'} \\
           	x & x
           	\arrow["{\equiv_0}", dashed, no head, from=1-1, to=1-2]
           	\arrow["{\equiv_0}", shift left, no head, from=2-1, to=1-1]
           	\arrow[shift right, from=2-1, to=1-1]
           	\arrow["{\equiv_0}", no head, from=2-1, to=2-2]
           	\arrow[dashed, from=2-2, to=1-2]
           	\arrow[shift left, from=2-2, to=3-2]
           	\arrow[from=3-1, to=2-1,ast->]
           	\arrow["{\equiv_1}", no head, from=3-1, to=3-2]
           	\arrow[shift left, from=3-2, to=2-2]
        \end{tikzcd}
        \\
        &
        \(\relR{y',z}\)
        &
        \(y' = z\), \(\relR{x,y'}\), \(\relR{x,y}\)
        &
        \(y'  = z\), \(\relR{x,y'}\), \(x = y\)
    \end{tabular}
    \caption{Case 1 (\(\relR{x,z}\)) in the proof of \conda{} in \zcref{lemma:diagram-equiv:tiptop-contains-topclust}.}\zlabel{fig:diagram-equiv:proof-tiptop-contains-topclust-condition-1-case-1}
\end{figure}

        If \(\relR{\vy['], \vz}\), then we are done, since by reflexivity of \(\relEquiv[0]\), we get \(\vz \relEquiv[0] \vz\).
        So suppose \(\neg\relR{\vy['], \vz}\).
        Since \(\relR[refl trans closure]{\vy['], \vx}\) and \(\relR{\vx, \vz}\), by weak transitivity, we get \(\vy['] = \vz\) and \(\vy['] \neq \vx\).
        Hence, \(\relR{\vx, \vy[']}\), \(\relR{\vy['],\vx}\), and \(\vy \relEquiv[0] \vy['] = \vz\).
        We distinguish two subcases.
        \begin{proofcaseitem}
            \item
            If \(\relR{\vx, \vy}\), then we can apply \condb{} to find \(\vz['] \in \fSpan[\sD]{\vy[']}\) with \(\vz \relEquiv[0] \vz[']\).
            \item
            If \(\neg\relR{\vx, \vy}\), then, since \(\relR[refl trans closure]{\vx, \vy}\), we get \(\vx = \vy\);
            so for \(\vz['] \coloneqq \vx\), we get \(\vz \relEquiv[0] \vz'\) and \(\relR{\vy['], \vz[']}\).
        \end{proofcaseitem}

        \item
        Suppose \(\neg\relR{\vx, \vz}\).
        By weak transitivity \(\vx = \vz\).
        Since \(\vy \in \sTopClust[\sU]\) and \(\relR{\vy, \vz}\), we get \(\relR{\vx, \vy}\).
        This is depicted in \zcref{fig:diagram-equiv:proof-tiptop-contains-topclust-condition-1-case-2}.
        We distinguish two subcases.

\begin{figure}
    \centering
    \begin{tabular}{c|cc}
\begin{tikzcd}
           	x \\
           	y & {y'} \\
           	x & x
           	\arrow[from=2-1, to=1-1]
           	\arrow["{\equiv_0}", no head, from=2-1, to=2-2]
           	\arrow[from=3-1, to=2-1]
           	\arrow["{\equiv_1}", no head, from=3-1, to=3-2]
           	\arrow[from=3-2, to=2-2,ast->]
        \end{tikzcd}
        &
\begin{tikzcd}
           	x & x \\
           	y & {y'} \\
           	x & x
           	\arrow["{\equiv_0}", dashed, no head, from=1-1, to=1-2]
           	\arrow[from=2-1, to=1-1]
           	\arrow["{\equiv_0}", no head, from=2-1, to=2-2]
           	\arrow[dashed, from=2-2, to=1-2]
           	\arrow[from=3-1, to=2-1]
           	\arrow["{\equiv_1}", no head, from=3-1, to=3-2]
           	\arrow[from=3-2, to=2-2]
        \end{tikzcd}
        &
\begin{tikzcd}
           	x & y \\
           	y & x \\
           	x & x
           	\arrow["{\equiv_0}", dashed, no head, from=1-1, to=1-2]
           	\arrow[from=2-1, to=1-1]
           	\arrow["{\equiv_0}", no head, from=2-1, to=2-2]
           	\arrow[dashed, from=2-2, to=1-2]
           	\arrow[from=3-1, to=2-1]
           	\arrow["{\equiv_1}", no head, from=3-1, to=3-2]
           	\arrow["{=}"', no head, from=3-2, to=2-2]
        \end{tikzcd}
        \\
        &
        \(\relR{x,y'}\)
        &
        \(x = y'\)
    \end{tabular}
    \caption{Case 2 (\(\vx = \vz\)) in the proof of \conda{} in \zcref{lemma:diagram-equiv:tiptop-contains-topclust}.}\zlabel{fig:diagram-equiv:proof-tiptop-contains-topclust-condition-1-case-2}
\end{figure}

        \begin{proofcaseitem}
            \item
            Suppose \(\relR{\vx, \vy[']}\).
            Since \(\vx \in \sTopClust[\sU]\), we get \(\relR{\vy['], \vx}\).
            So \(\vz['] \coloneqq \vx\) closes the diagram.
            \item
            Suppose \(\neg\relR{\vx, \vy[']}\).
            Since \(\relR[refl trans closure]{\vx, \vy[']}\), we obtain \(\vx = \vy[']\), so \(\vz['] \coloneqq \vy\) closes the diagram.
        \end{proofcaseitem}
    \end{proofcases}

    For \condb{} let \(\vz \in \fSpan[\sD]{\vy}\), \(\vw \in \fSpan[\sD]{\vz}\) and \(\vz['] \in \fSpan[\sD]{\vy[']}\) with \(\vz['] \relEquiv[0] \vz \relEquiv[0] \vw\).
    We want to show that there exists \(\vw['] \in \fSpan[\sD]{\vz[']}\) such that \(\vw \relEquiv[0] \vw[']\), i.e., we want to complete the diagram on the left in \zcref{fig:diagram-equiv:proof-tiptop-contains-topclust-condition-2}.
We distinguish four cases which are depicted diagrammatically in \zcref{fig:diagram-equiv:proof-tiptop-contains-topclust-condition-2}.
    \begin{proofcases}
        \item
        If \(\relR{\vx, \vz}\) and \(\relR{\vx, \vz[']}\), then we can apply \condb{} directly to find \(\vw['] \in \fSpan[\sD]{\vz[']}\) with \(\vw \relEquiv[0] \vw[']\).

        \item
        If \(\neg\relR{\vx, \vz}\) and \(\relR{\vx, \vz[']}\), then, by weak transitivity, \(\vx = \vz\), so \(\vx \relEquiv[0] \vw\). Moreover, since \(\vx \in \sTopClust[\sU]\), we also get \(\relR{\vz['], \vx}\), so \(\vw['] \coloneqq \vx\) closes the diagram.

        \item
        If \(\relR{\vx, \vz}\) and \(\neg\relR{\vx,\vz[']}\), then, by weak transitivity, \(\vx = \vz[']\), so \(\vw['] \coloneqq \vz\) closes the diagram.

        \item
        If \(\neg\relR{\vx, \vz}\) and \(\neg\relR{\vx, \vz[']}\), then, by weak transitivity, \(\vx = \vz = \vz[']\), so \(\vw['] \coloneqq \vw\) closes the diagram.\qedhere
    \end{proofcases}
\end{proof}

\begin{figure}
    \centering
    \begin{tabular}{c|cccc}
\begin{tikzcd}
    	   {w} & \\
    	   z & {z'} \\
    	   y & {y'} \\
    	   x & x
\arrow["{\equiv_0}", shift left, no head, from=2-1, to=1-1]
    	   \arrow[shift right, from=2-1, to=1-1]
    	   \arrow["{\equiv_0}", no head, from=2-1, to=2-2]
\arrow[from=3-1, to=2-1]
    	   \arrow["{\equiv_0}", no head, from=3-1, to=3-2]
    	   \arrow[from=3-2, to=2-2]
    	   \arrow[from=4-1, to=3-1,ast->]
    	   \arrow["{\equiv_1}", no head, from=4-1, to=4-2]
    	   \arrow[from=4-2, to=3-2,ast->]
        \end{tikzcd}
        &
\begin{tikzcd}
           	{w} & {w'} \\
           	z & {z'} \\
           	y & {y'} \\
           	x & x
           	\arrow["{\equiv_0}", dashed, no head, from=1-1, to=1-2]
           	\arrow["{\equiv_0}", shift left, no head, from=2-1, to=1-1]
           	\arrow[shift right, from=2-1, to=1-1]
           	\arrow["{\equiv_0}", no head, from=2-1, to=2-2]
           	\arrow[dashed, from=2-2, to=1-2]
           	\arrow[from=3-1, to=2-1]
           	\arrow["{\equiv_0}", no head, from=3-1, to=3-2]
           	\arrow[from=3-2, to=2-2]
           	\arrow[curve={height=-12pt}, from=4-1, to=2-1]
           	\arrow[from=4-1, to=3-1,ast->]
           	\arrow["{\equiv_1}", no head, from=4-1, to=4-2]
           	\arrow[curve={height=12pt}, from=4-2, to=2-2]
           	\arrow[from=4-2, to=3-2,ast->]
        \end{tikzcd}
        &
\begin{tikzcd}
           	{w} & x \\
           	x & {z'} \\
           	y & {y'} \\
           	x & x
           	\arrow["{\equiv_0}", dashed, no head, from=1-1, to=1-2]
           	\arrow["{\equiv_0}", shift left, no head, from=2-1, to=1-1]
           	\arrow[shift right, from=2-1, to=1-1]
           	\arrow["{\equiv_0}", no head, from=2-1, to=2-2]
           	\arrow[dashed, from=2-2, to=1-2]
           	\arrow[from=3-1, to=2-1]
           	\arrow["{\equiv_0}", no head, from=3-1, to=3-2]
           	\arrow[from=3-2, to=2-2]
           	\arrow[from=4-1, to=3-1,ast->]
           	\arrow["{\equiv_1}", no head, from=4-1, to=4-2]
           	\arrow[curve={height=12pt}, from=4-2, to=2-2]
           	\arrow[from=4-2, to=3-2,ast->]
        \end{tikzcd}
        &
\begin{tikzcd}
           	{w} & z \\
           	z & x \\
           	y & {y'} \\
           	x & x
           	\arrow["{\equiv_0}", dashed, no head, from=1-1, to=1-2]
           	\arrow["{\equiv_0}", shift left, no head, from=2-1, to=1-1]
           	\arrow[shift right, from=2-1, to=1-1]
           	\arrow["{\equiv_0}", no head, from=2-1, to=2-2]
           	\arrow[dashed, from=2-2, to=1-2]
           	\arrow[from=3-1, to=2-1]
           	\arrow["{\equiv_0}", no head, from=3-1, to=3-2]
           	\arrow[from=3-2, to=2-2]
           	\arrow[curve={height=-12pt}, from=4-1, to=2-1]
           	\arrow[from=4-1, to=3-1,ast->]
           	\arrow["{\equiv_1}", no head, from=4-1, to=4-2]
           	\arrow[from=4-2, to=3-2,ast->]
        \end{tikzcd}
        &
\begin{tikzcd}
           	{w} & w \\
           	x & x \\
           	y & {y'} \\
           	x & x
           	\arrow["{\equiv_0}", dashed, no head, from=1-1, to=1-2]
           	\arrow["{\equiv_0}", shift left, no head, from=2-1, to=1-1]
           	\arrow[shift right, from=2-1, to=1-1]
           	\arrow["{\equiv_0}", no head, from=2-1, to=2-2]
           	\arrow[dashed, from=2-2, to=1-2]
           	\arrow[from=3-1, to=2-1]
           	\arrow["{\equiv_0}", no head, from=3-1, to=3-2]
           	\arrow[from=3-2, to=2-2]
           	\arrow[from=4-1, to=3-1,ast->]
           	\arrow["{\equiv_1}", no head, from=4-1, to=4-2]
           	\arrow[from=4-2, to=3-2,ast->]
        \end{tikzcd}
        \\
        &
        \(\relR{\vx,\vz}\), \(\relR{\vx,\vz[']}\)  &\(\vx=\vz\), \(\relR{\vx,\vz[']}\) & \(\relR{\vx,\vz}\), \(\vx = \vz[']\) & \(\vx = \vz = \vz[']\)
    \end{tabular}
    \caption{The four cases in the proof of \condb{} in \zcref{lemma:diagram-equiv:tiptop-contains-topclust}.}\zlabel{fig:diagram-equiv:proof-tiptop-contains-topclust-condition-2}
\end{figure}

In conclusion, we find the following criterion for the definability of \(\sTopClust\).

\begin{lemma}\zlabel{lemma:top-cluster-definable}
    Suppose \(\lthProp\) is finite, \(\sU\) is finite, and \(\frmM\) is collapsed relative to \(\sU\).  Then \(\sTopClust[\sU]\) is finite.
    If, additionally, every point in \(\sU\) is definable, then  each point in \(\sTopClust[\sU]\) is definable.
\end{lemma}
\begin{proof}
By \zcref{lemma:diagram-equiv:equiv-2:bisim}, \({\relEquiv[d]} \sUnion {\relDiag[\sU]}\) is a bisimulation equivalence on \(\frmM\) restricted to the upset \(\sU \sUnion {\relEquiv[d,dom]}\).
    Since \(\frmM\)  is collapsed relative to \(\sU\), it follows that \({\relEquiv[d]}\) is the diagonal on \({\relEquiv[d,dom]}\).
    Therefore, by \zcref{lemma:diagram-equiv:finite-definable-equiv-classes}, \(\relEquiv[d,dom]\) is finite and under the extra assumption each point in it is definable.
    Moreover, by \zcref{lemma:diagram-equiv:tiptop-contains-topclust}, \(\sTopClust[\sU] \subseteq {\relEquiv[d,dom]}\), so the same is true for \(\sTopClust[\sU]\).
\end{proof}

Applied to the selection construction, we obtain the following result for the finite stages.

\begin{proposition}\zlabel{lemma:selective-filtration:finite-stage-finite-quotient-head}
    If the set of atomic propositions \(\lthProp\) is finite, then, in \zcref{constr:selective-filtration}, \(\fq[\nn]{\sU[\nn]}\) is finite and \(\sX[\nn]\) is definable in \(\frmM\) for each \(\nn \in \Naturals\).
\end{proposition}
\begin{proof}
    We prove by induction on \(\nn \in \Naturals\) that \(\fq[\nn]{\sU[\nn]}\) is finite, and  that \(\sU[\nn]\), \(\sX[\nn]\), and the elements of \(\sU[\nn + 1, quotient={\relBisim[\nn + 1]}]\) are definable in \(\frmM\).
	For \(\nn = 0\) this holds vacuously.

    Suppose that \(\sU[tilde, \nn] = \fq[\nn]{\sU[\nn]}\) is finite, and that \(\sU[\nn]\), \(\sX[\nn]\), and the elements of \(\sU[\nn, quotient={\relBisim[\nn]}]\) are definable in \(\frmM\).
    Let \(\frmM[\nn, upper={+}]\) be the extension of \(\frmM[\nn]\) with, for every \(\sA \in \sU[\nn, quotient={\relBisim[\nn]}]\), a fresh atomic proposition \(\lap[\sA]\) which evaluates to \(\sA\).
    Clearly, \(\frmM[\nn, upper={+}]\) is a definable submodel of a definable variant of \(\frmM\).
    Moreover, since \(\relBisim[\nn]\) is a bisimulation equivalence on \(\frmM[\nn]\) and each \(\sA \in \sU[\nn, quotient={\relBisim[\nn]}]\) is by definition closed under it, \(\relBisim[\nn]\) is a bisimulation equivalence on \(\frmM[\nn, upper={+}]\).
    Write \(\frmM[tilde, \nn, upper={+}]\) for the quotient \(\frmM[\nn, upper={+}, quotient={\relBisim[\nn]}]\).

    By definition, every point in \(\fq[\nn]{\sU[\nn]}\) is definable in \(\frmM[\nn, tilde, upper={+}]\), \(\frmM[\nn, tilde, upper={+}]\) is collapsed relative to \(\fq[\nn]{\sU[\nn]}\), and, by the induction hypothesis, \(\fq[\nn]{\sU[\nn]}\) is finite.
    Therefore, by \zcref{lemma:top-cluster-definable}, \(\sTopClust[\fq[\nn]{\sU[\nn]}]\) is finite and every point in it is definable in \(\frmM[\nn, tilde, upper={+}]\).
    Hence, \(\sT[\nn]\) and the \(\relBisim[\nn]\) equivalence classes in it are definable in \(\frmM[\nn, upper={+}]\), and hence in \(\frmM\).
    By \zsubref{lemma:selective-filtration:basic-props}{lemma:selective-filtration:bisim-next-head}, each \(\relBisim[\nn + 1]\) equivalence class in \(\sU[\nn + 1]\) is the union of \(\relBisim[\nn]\) equivalence classes, hence there are only finitely many such equivalence classes which all are definable.
    In particular, \(\fq[\nn + 1]{\sU[\nn + 1]}\) is finite.
    Finally, by \zcref{lemma:eliminability:definable}, it follows that \(\ElimCrit[m=\frmM[\nn], u=\sU[\nn + 1]]\) is definable in  \(\frmM\), and the definability of \(\sX[\nn + 1]\) follows.
\end{proof}
 
    \section{Fine's Theorems}\zlabel{sec:Fines-thms}
    With the necessary theory developed, we now prove the main results of the paper. We start with the Selection Theorem.

\begin{maintheoremagain}\zlabel{mthm:selection-again}
    Let \(\frmM\) be a weakly transitive Kripke model and \(\lthSigma\) a finite subformula-closed set of formulas.
	Then there exists a definable \(\lthSigma\)-sub-p-morphism from \(\frmM\) onto a finite model.
\end{maintheoremagain}
\begin{proof}
    We apply \zcref{constr:selective-filtration} to \(\frmM\). It suffices to show that the construction terminates at stage \(\nm\coloneqq 2^{\lthSigma[card]}\). Then, by \zcref{lemma:selection-construction-sigma-subset} together with \zcref{lemma:selective-filtration:finite-stage-finite-quotient-head}, it follows that \(\fq[\nm]\) is a \(\lthSigma\)-sub-p-morphism onto the finite model \(\frmM[tilde,\nm]\).

    So let \(\nn \in \Naturals\).
    Then, by \zcref{lemma:selective-filtration:finite-stage-finite-quotient-head}, \(\fq[\nn]{\sU[\nn]}\) is finite, and hence, by \zcref{lemma:selective-filtration:finite-width-top-implies-covering}, \(\sT[\nn]\) covers \(\sX[\nn] \sMinus \sU[\nn]\).
	In particular the construction does not get stuck at any finite stage. Suppose for a contradiction that the construction has not terminated at stage \(\nm\).
	Then \(\sT[\nm] \neq \sEmpty\), so there exists \(\vx[\nm] \in \sT[\nm]\), and
	inductively, we find \(\vx[\nm], \dots, \vx[0]\) such that \(\vx[\nii] \in \sT[\nii]\) and \(\vx[\nii] \in \relR[\nm]{\vx[\nii + 1]}\) for each \(\nii \leq \nm\).
	By \zsubref{lemma:selective-filtration:basic-props}{lemma:selective-filtration:related-ty-equal-points-same-top}, the points \(\vx[\nm], \dots, \vx[0]\)  have pairwise distinct \(\lthSigma\)-types.
	But there are only \(\nm = 2^{\lthSigma[card]}\) distinct \(\lthSigma\)-types, a contradiction.
\end{proof}

Now suppose \(\logLambda\) is a strongly cofinal subframe logic extending \(\logKWeakTrans\), and \(\lphi \notin \logLambda\).
Then there exists a weakly transitive  \(\logLambda\)-model \(\frmM\) refuting \(\lphi\).
Let \(\lthSigma\) be the set of subformulas of \(\lphi\).
By the previous theorem, we obtain a definable \(\lthSigma\)-sub-p-morphism from \(\frmM\) onto a finite model \(\frmN\).
But every \(\lthSigma\)-subset of \(\frmM\) is strongly cofinal, so \(\frmN\) is a \(\logLambda\)-model.
Moreover, by \zcref{prop:Sigma-sub-p-morphism-preserves-truth}, \(\frmN\) still refutes \(\lphi\).
Thus, we get the claimed \fmp{} result:

\begin{maintheoremagain}\zlabel{mthm:fmp-again}
    Every strongly cofinal subframe logic extending \(\logKWeakTrans\) has the \fmp{}.
\end{maintheoremagain}
 
Next, we prove Fine's Finite Width Theorem for \(\logKWeakTrans\), using our selection construction with an infinite filtration set.
We will show that the resulting model is atomic, and finally that the underlying Kripke frame still has the same logic.

Let \(\logForm\) be the set of all modal formulas over finitely many atomic propositions \(\lthProp\).
We apply our generalized selection construction to a given weakly transitive model \(\frmM\) of finite width with \(\lthSigma \coloneqq \logForm\).
By \zcref{cor:finite-width-termination}, the construction terminates at some stage \(\odalpha\) and the resulting model \(\frmM[tilde, \odalpha]\) has the finite cover property.
Moreover, by \zcref{lemma:selective-filtration:head-no-eliminable-points}, it is free of \(\logForm\)-eliminable points.
For models with these properties, we derive a Hennessy-Milner-like theorem.

\begin{theorem}\zlabel{lemma:finite-width-thm:Hennessy-Milner-like}
    Suppose \(\lthProp\) is finite.
    Let \(\frmM = \structuple{\sX, \relR, \fVal}\) be a weakly transitive model free of \(\logForm\)-eliminable points, such that every point generated subframe has the finite cover property.
    Then \(\relEquiv[ty=\logForm]\) and bisimilarity on \(\frmM\) coincide.
\end{theorem}

In the weakly transitive case, image-finiteness is strictly stronger than point generated subframes having the finite cover property.
Relaxing this hypothesis compared to the Hennessy-Milner Theorem comes at the cost of requiring the frame to be free of \(\logForm\)-eliminable points.

\begin{proof}
	It suffices to prove that \(\relEquiv[ty=\logForm]\) is a bisimulation equivalence.
	Let \(\vx, \vx['], \vy \in \sX\) be such that \(\vy \in \relR{\vx}\) and \(\vx \relEquiv[ty=\logForm] \vx[']\).
	For \(\nn \in \Naturals\), let \(\lphi[\nn]\) be a formula of modal depth at most \(\nn\) such that
	\begin{equation*}
		\sem{\lphi[\nn]} = \sSetIntersection\Set{ \sem{\lphi} }[ \vy \in \sem{\lphi}, \text{\(\lphi\) is of modal depth at most \(\nn\)} ] .
	\end{equation*}
	Then \(\vy \in \sem{\lphi[\nn]}\) and therefore \(\vx \in \sem{\lDiamond\lphi[\nn]}\). Since \(\vx['] \relEquiv[ty=\logForm] \vx\), there exists \(\vy[\nn, '] \in \relR{\vx[']}\) with \(\vy[\nn,'] \in \sem{\lphi[\nn]}\).
	For each \(\nn \in \Naturals\), pick such a \(\vy[\nn, ']\), and define \(\sY['] \coloneqq \Set{\vy[\nn, ']}[\nn \in \Naturals]\).
    Suppose for a contradiction that there does not exist \(\vy['] \in \sY[']\) with \(\vy[']\relEquiv[ty=\logForm]\vy\).
    Then for each \(\nn \in \Naturals\), there exists an \(\nm\in \Naturals\) such that \(\vy[',\nn]\notin \sem{\lphi[\nm]}\). Since \(\lthProp\) is finite, every cluster is finite up to bisimilarity, so \(\sY[']\) contains points in infinitely many clusters.
    Hence, by the finite cover property for the subframe generated by the point \(\vx[']\), \(\sY[']\) contains an infinite strictly decreasing chain, so
    there is  an infinite set \(\sI \sSubsetEq \Naturals\) such that for \(\nn, \nm \in \sI\), if \(\nn < \nm\) then \(\relR{\vy[\nm, '], \vy[\nn, ']}\) and \(\vy[\nn, '] \notin \sem{\lphi[\nm]}\).

	For \(\nn \in \sI\), let \(\fs{\nn} \coloneqq \min\Set{\nm \in \sI}[\nn < \nm]\).
	Then for each \(\nn \in \sI\), \(\vy[{\fs{\fs{\nn}}}, '] \in \sem{ \lDiamond{\lphi[\nn] \land \lnot\lphi[\fs{\nn}]} }\) and, since \(\lDiamond{\lphi[\nn] \land \lnot\lphi[\fs{\nn}]}\) has modal depth \(\fs{\nn} + 1 \leq \fs{\fs{\nn}}\), we get \(\vy \in \sem{ \lDiamond{\lphi[\nn] \land \lnot\lphi[\fs{\nn}]} }\).
	Therefore, there exists \(\vy[\nn] \in \relR{\vy}\) with \(\vy[\nn] \in \sem{\lphi[\nn] \land \lnot\lphi[\fs{\nn}]}\).
	Since the cluster of \(\vy\) is finite up to bisimilarity, the set \(\sJ = \Set{\nn \in \sI}[\vy \notin \relR{\vy[\nn]}]\) is infinite.
Finally, let \(\lphi \in \logForm\) and suppose \(\vy \in \sem{\lphi}\).
	Let \(\nn\) be the modal depth of \(\lphi\) and \(\nm \in \sJ\) with \(\nn < \nm\). The point  \(\vy[\nm]\) is a strict successor of \(\vy\) in \(\sem{\lphi[\nm]} \sSubsetEq \sem{\lphi[\nn]} \sSubsetEq \sem{\lphi}\), i.e., \(\vy[\nm] \in \sem{\lphi}\).
	Hence, since \(\lphi\) was arbitrary, \(\vy\) is \(\logForm\)-eliminable, contradicting the hypotheses.
\end{proof}

Since the model \(\frmM[tilde, \odalpha]\) resulting from the selection is collapsed relative to its domain, by the previous Theorem, the relation \(\relEquiv[ty=\logForm]\) on \(\frmM[tilde, \odalpha]\) is the diagonal, i.e.\@ \(\frmM[tilde, \odalpha]\) is differentiated.
The following lemma strengthens this to atomicity.

\begin{lemma}\zlabel{lemma:finite-width-thm:atomicity}
    Suppose that \(\lthProp\) is finite.
    Let \(\frmM = \structuple{\sX, \relR, \fVal}\) be a weakly transitive differentiated model that has the finite cover property and is free of \(\logForm\)-eliminable points.
    Then \(\frmM\) is atomic.
\end{lemma}

\begin{proof}
	Let \(\vx \in \sX\) and define \(\fClust{\vx}\coloneqq\relR{\vx} \cap\relR[converse]{\vx} \sUnion \Set{\vx}\).
	We prove that \(\Set{\vx}\) is definable.
	Since \(\vx\) is not eliminable there exists \(\lphi[0]\) such that \(\vx\) is maximal in \(\sem{\lphi[0]}\).
    By assumption, there exists a finite cover \(\sC\) of \(\sem{\lphi[0]} \sMinus \fClust{\vx}\).
	Since \(\lthProp\) is finite and \(\frmM\) is a weakly transitive differentiated model, \(\fClust{\vx}\) is finite.
    Thus, by differentiation, there exists \(\lphi[\vx]\) such that \(\fClust{\vx} \sIntersection \sem{\lphi[\vx]} = \Set{\vx}\) and for each \(\vz \in \sC\), there exists \(\lphi[\vz]\) such that \(\fClust{\vx} \sSubsetEq \sem{\lphi[\vz]}\) but \(\vz \notin \sem{\lphi[\vz]}\).
	Consider the definable set
	\begin{equation*}
		\sY \coloneqq \sem{\lphi[0]} \sIntersection \sem{\lphi[\vx]} \sIntersection \sSetIntersection\Set{ \sem{\lBox[refl closure]{\lphi[0] \limplies \lphi[\vz]}} }[\vz \in \sC].
	\end{equation*}
We claim that \(\Set{\vx} =\sY\).

    For the left-to-right inclusion, note that \(\vx \in \sem{\lphi[0]} \sIntersection \sem{\lphi[\vx]}\).
    Moreover, since \(\vx\) is maximal in \(\sem{\lphi[0]}\), for  \(\vz \in \sC\) and \(\vx['] \in \relR[refl trans closure]{\vx}\) with \(\vx['] \in \sem{\lphi[0]}\), we get \(\vx['] \in \fClust{\vx} \sSubsetEq \sem{\lphi[\vz]}\).
    Hence \(\vx \in \sY\).

    For the converse, suppose that \(\vx['] \in \sY \sMinus \Set{\vx}\).
	Then \(\vx['] \in \sem{\lphi[0]} \sMinus \fClust{\vx}\)  and \(\vx[']\) is covered by some \(\vz \in \sC\).
	By weak transitivity, \(\vz \in \relR[refl trans closure]{\vx[']}\).
	We know by assumption that \(\vx['] \in \sem{\lBox[refl closure]{\lphi[0] \limplies \lphi[\vz]}}\).
	Hence, \(\vz \in \sem{\lphi[0] \limplies \lphi[\vz]}\) and \(\vz \in \sC \sSubsetEq \sem{\lphi[0]}\), but, by definition of \(\lphi[\vz]\), \(\vz \notin \sem{\lphi[\vz]}\), contradiction.
\end{proof}

We conclude that the model \(\frmM[tilde, \odalpha]\) resulting from the selection construction is atomic.

\begin{proposition}\zlabel{thm:finite-width-thm:selection-result-atomic}
	Suppose that \(\lthProp\) is finite, \(\lthSigma = \logForm\), and in \zcref{constr:selective-filtration} \(\frmM\) is of finite width.
	Then the construction terminates, say at stage \(\odalpha\), and \(\frmM[\odalpha, tilde]\) is atomic, has the finite cover property, and validates exactly the same formulas as \(\frmM\).
\end{proposition}

In particular, any logic of finite width, is complete w.r.t.\@ such conversely pre-well-founded and atomic models.
It now suffices to prove persistence for these models.
We first derive an auxiliary lemma.

\begin{lemma}\zlabel{lemma:finite-width-thm:atomic-finite-restricted-sub-p-morphism}
    Let \(\lthSigma\) be finite and let \(\frmM\) and \(\frmM[tilde]\) be models such that \(\frmM\) has the finite cover property and \(\frmM[tilde]\) is finite, and let \(\ff\) be \(\lthSigma\)-sub-p-morphism from \(\frmM\) to \(\frmM[tilde]\).
    Then there exists a restriction \(\fg\) of \(\ff\) with finite domain which is still a \(\lthSigma\)-sub-p-morphism from \(\frmM\) to \(\frmM[tilde]\).
\end{lemma}
\begin{proof}
    Let \(\vx[tilde]\) be a point of \(\frmM[tilde]\).
    By the finite cover property, there exists a finite cover \(\sC[\vx[tilde]]\) of \(\ff[preimage]{\vx[tilde]}\). Without loss of generality, we may assume that for each \(\vx \in \sC[\vx[tilde]]\), if \((\relR{\vx} \sIntersection \relR[converse]{\vx})\sIntersection\ff[preimage]{\vx[tilde]} \neq \sEmpty\), then \((\relR{\vx} \sIntersection \relR[converse]{\vx})\sIntersection\sC[\vx[tilde]] \neq \sEmpty\), that is, if for \(\vx \in \sC[\vx[tilde]]\) the set \(\ff[preimage]{\vx[tilde]}\) contains a point in the cluster of \(\vx\), then also \(\sC[\vx[tilde]]\) contains a point in the cluster of $\vx$.

    For every \(\lphi \in \lthSigma\) and \(\vx \in \sC[\vx[tilde]]\), if there exists \(\vy \in \relR{\vx} \sIntersection \relR[converse]{\vx}\) with \(\vx \in \sem{\lphi} \iff \vy \in \sem{\lphi}\), then there exists \(\vy['] \in \fSpan[\ff[dom]]{\vx}\) with \(\vx \in \sem{\lphi} \iff \vy['] \in \sem{\lphi}\).
    Whenever this is the case, pick some such \(\vy[']\) and add it to \(\sC[\vx[tilde]]\) to obtain \(\sC[',\vx[tilde]]\) which is finite, since \(\lthSigma\) is finite.
    Define \(\fg\) to be the restriction of \(\ff\) to the union of the finitely many \(\sC[',\vx[tilde]]\).

    It is easy to check that \(\fg[dom]\) is a finite \(\lthSigma\)-subset of \(\frmM\),
and monotonicity of \(\fg\) is inherited from \(\ff\).
    We check the back-condition for \(\fg\).
    Let \(\vx \in \fg[dom]\), \(\vx[tilde] \coloneqq \fg{\vx}\) and \(\vy[tilde]\) a successor of \(\vx[tilde]\).
    By the back-condition of \(\ff\), there exists \(\vy \in \relR{\vx} \sIntersection \ff[preimage]{\vy[tilde]}\).
    Since \(\sC[\vy[tilde]]\) covers \(\ff[preimage]{\vy[tilde]}\), there is \(\vz \in \relR[ refl trans closure]{\vy} \sIntersection \sC[\vy[tilde]]\).
    If \(\vz \in \relR{\vx}\), then we are done.
    Otherwise \(\vz \notin \relR{\vx}\), and,
    by weak transitivity, \(\vx = \vz\), so \(\vx[tilde] = \vy[tilde]\) and  \(\vy \in (\relR{\vx} \sIntersection \relR[converse]{\vx})\sIntersection\ff[preimage]{\vx[tilde]}\).
    Hence, there exists $\vy['] \in (\relR{\vx} \sIntersection \relR[converse]{\vx})\sIntersection \sC[\vx[tilde]]$.
\end{proof}

\begin{proposition}\zlabel{thm:finite-width-thm:persistence}
	Let \(\logLambda\) be an extension of \(\logKWeakTrans\) and \(\frmM\) an atomic \(\logLambda\)-model on finitely many propositional variables \(\lthProp\) with the finite cover property.
	Then the frame \(\frF = \structuple{\sX, \relR}\) underlying \(\frmM\) is a \(\logLambda\)-frame.
\end{proposition}

\begin{proof}
    Let \(\frmN\) be a model on \(\frF\) and let \(\lthSigma[\nn]\) be the set of formulas of modal depth at most \(\nn\) up to logical equivalence. It is enough to show that for each \(\nn \in \Naturals\), \(\frmN\) satisfies all the formulas in \(\lthSigma[\nn] \sIntersection \logLambda\).
    So let \(\nn \in \Naturals\), and write \(\lthSigma \coloneqq \lthSigma[\nn]\).
	By \zcref{mthm:selection-again} and \zcref{lemma:finite-width-thm:atomic-finite-restricted-sub-p-morphism}, there exists a finite \(\lthSigma\)-subset \(\sY\) of \(\frmN\).
	Write \(\frmN[restrict=\sY]\) for the restriction of \(\frmN\) to \(\sY\).

	Let \(\vx \in \sX\).
	Since \(\sY\) is a \(\lthSigma\)-subset of \(\frmN\) and \(\lthSigma\) is Boolean closed (up to equivalence), there exists \(\vy \in \relR[refl trans closure]{\vx} \sIntersection \sY\) with \(\fTy[\lthSigma]{\frmN, \vx} = \fTy[\lthSigma]{\frmN, \vy}\).
	Then by \zcref{lemma:main-property-Sigma-submodel}
	\begin{equation}\label{eq:finite-width-thm:persistence:h-x-defining-property}
		\ForAll{\lDiamond\lpsi \in \lthSigma}{
		    \vx \in \relR[converse]{\sY \sIntersection \sem{\lpsi}[{\frmN}]}
			\miff
			\vy \in \relR[converse]{\sY \sIntersection \sem{\lpsi}[{\frmN}]}
}.
	\end{equation}
	Fix a total order \(\LessEq\) on \(\sY\).
    Define the map \(\Function{\fh}{\sX}{\sY}\) such that \(\fh{\vx}\) is the \(\LessEq\)-minimal \(\vy \in  \relR{\vx} \sIntersection \sY\) for which \zcref{eq:finite-width-thm:persistence:h-x-defining-property} holds, if \(\vx \notin \sY\), and \(\fh{\vx} \coloneqq \vx\) otherwise.

	Since \(\sY\) is finite and \(\frmM\) is atomic, by \zcref{lemma:span-equiv-definable}, \(\relEquiv[span=\sY]\) is definable in \(\frmM\), and moreover has finitely many equivalence classes.
	Note that \(\fh[preimage]{\vy}\) is a union of equivalence classes of \(\relEquiv[span=\sY]\) restricted to \(\sX\sMinus\sY\) and \(\Set{\vy}\), hence definable in \(\frmM\).

	Let \(\fVal\) be the valuation of \(\frmN\), and define a new valuation \(\fVal[']\) on \(\frF\) by \(\fVal[']{\lap} \coloneqq \fh[preimage]{\fVal{\lap}\sIntersection \sY}\).
	Write \(\frmM[']\) for the resulting model.
	Since \(\sY\) is finite, \(\fVal[']{\lap}\) is definable in \(\frmM\), so \(\frmM[']\) is a definable variant of \(\frmM\), hence a \(\logLambda\)-model.
	We show that \(\frmM[']\) satisfies the same \(\lthSigma\)-formulas as \(\frmN\).

	Note that by definition \(\frmM[{', restrict=\sY}] = \frmN[restrict=\sY]\).
	By \zcref{prop:Sigma-sub-p-morphism-preserves-truth}, \(\sY \sIntersection \sem{\lpsi}[{\frmN}] = \sem{\lpsi}[{\frmN[{restrict=\sY}]}]\).
    Hence, by definition of \(\fh\), we get for all \(\vx \in \sX\), \(\fh{\vx} \in \relR[refl trans closure]{\vx}\) and
\begin{equation}\label{eq:h-property}
	 	\ForAll{\lDiamond\lpsi \in \lthSigma}{
	 	    \vx \in \relR[converse]{\sem{\lpsi}[{\frmM[{', restrict=\sY}]}]}
	 		\miff
	 		\fh{\vx} \in \relR[converse]{\sem{\lpsi}[{\frmM[{', restrict=\sY}]}]}
}.
	  \end{equation}

	We prove by induction on formulas \(\lpsi \in \lthSigma\), that \(\vx \in \sem{\lpsi}[\frmM[']]\) iff \(\fh{\vx} \in \sem{\lpsi}[\frmM[{', restrict=\sY}]]\), for all \(\vx \in \sX\).
	For atomic propositions this is trivial, and the steps for meet and negation are standard.
	For the diamond step, let \(\lDiamond\lpsi \in \lthSigma\) and \(\vx \in \sX\).
Suppose \(\vx \in \sem{\lDiamond\lpsi}[\frmM[']]\) and find a successor \(\vx['] \in \sem{\lpsi}[\frmM[']]\).
	By the induction hypothesis \(\fh{\vx[']} \in \sem{\lpsi}[\frmM[{', restrict=\sY}]]\).
	If \(\fh{\vx[']} \in \relR{\vx}\) then, by \zcref{eq:h-property}, we get \(\fh{\vx} \in \sem{\lDiamond\lpsi}[{\frmM[{', restrict=\sY}]}]\) and we are done.
	So suppose \(\fh{\vx[']} \notin \relR{\vx}\).
    Then, since \(\vx['] \in \relR{\vx}\) and \(\fh{\vx[']} \in \relR[refl trans closure]{\vx[']}\), by weak transitivity,  it follows that \(\vx = \fh{\vx[']}\), yielding  \(\vx \neq \vx[']\) and \(\vx \in \relR{\vx[']}\).
	Moreover, \(\vx \in \sY\) and hence \(\fh{\vx} = \vx\).
	Now \(\vx['] \in \relR[converse]{\sem{\lpsi}[{\frmM[{', restrict=\sY}]}]}\), so, by \zcref{eq:h-property}, we get \(\vx = \fh{\vx[']} \in \relR[converse]{\sem{\lpsi}[{\frmM[{', restrict=\sY}]}]}\), yielding \(\vx \in \sem{\lDiamond\lpsi}[{\frmM[{', restrict=\sY}]}]\).

	Conversely, suppose \(\fh{\vx} \in \sem{\lDiamond\lpsi}[\frmM[{', restrict=\sY}]]\).
	By \zcref{eq:h-property}, there exists \(\vy \in \relR{\vx} \sIntersection \sem{\lpsi}[{\frmM[{', restrict=\sY}]}]\).
	But, since \(\vy \in \sY\) we have \(\fh{\vy} = \vy\). Hence, by the induction hypothesis, \(\vy \in \sem{\lpsi}[\frmM[']]\), yielding \(\vx \in \sem{\lDiamond\lpsi}[\frmM[']]\).

	Now suppose \(\frmN\) refutes \(\lphi \in \lthSigma[\nn]\).
	Since \(\sY\) is a \(\lthSigma[\nn]\)-subset in \(\frmN\), there is \(\vy \in \sY\) such that \(\vy \in \sem{\lnot\lphi}[\frmN]\).
	By \zcref{prop:Sigma-sub-p-morphism-preserves-truth}, we have \(\vy \in \sem{\lnot\lphi}[\frmN[{restrict=\sY}]] = \sem{\lnot\lphi}[\frmM[{', restrict=\sY}]]\), and since \(\fh{\vy} = \vy\) it follows, by the inductive proof above, that \(\vy \in \sem{\lnot\lphi}[\frmM[']]\).
	Since \(\frmM[']\) is a \(\logLambda\)-model, we conclude \(\lphi \notin \logLambda\).
\end{proof}

Combining \zcref{thm:finite-width-thm:persistence,thm:finite-width-thm:selection-result-atomic}, we obtain the Finite Width Theorem for \(\logKWeakTrans\):

\begin{maintheoremagain}\zlabel{mthm:finite-width-again}
	Every logic of finite width extending \(\logKWeakTrans\) is complete w.r.t.\@ its conversely pre-well-founded frames, and is therefore, in particular, Kripke complete.
\end{maintheoremagain}
 
    \section{Concluding Remarks}\zlabel{sec:conclusion}
    We have generalized Fine's Iterative Selection Method to the weakly transitive setting and to infinite filtration sets \(\lthSigma\), and have shown that using the set of all formulas on finitely many atomic propositions as the filtration set provides an alternative proof of Fine's Finite Width Theorem.
A natural next step is to look for further applications of our construction.
Our construction might help to drop canonicity assumptions in existing selection constructions on canonical models of subframe logics, e.g., \cite{BaltagBezhanishviliFernandez2023-topological-mu-calclulus-completeness-and-decidability,KudinovShapirovsky2025-two-types-of-filtrations-for-wK4-and-relatives}.
On the other hand, infinite filtration sets might be useful for non-subframe logics.
Another direction is to try to generalize it beyond weakly transitive logics, for example, to \(\nn\)-transitive logics.
A limitative result for the latter was given by Wolter~\cite[Section~4.3]{Wolter1993-lattices-of-modal-logics} who constructed examples of subframe logics extending \(\logK\) that lack the \fmp{}.

Zakharyaschev's theory of canonical formulas forms a major application of Fine's Selection Construction, and could, using our result, be generalized to \(\logKWeakTrans\).
This theory has already been generalized using algebraic methods \cite{BezhanishviliBezhanishvili2012-canonical-formulas-wK4}, but the frame theoretic perspective might provide an easier path forward for handling the quasi-normal case, cf.~\cite[Section~7]{Zakharyaschev1996-canonical-formulas-for-K4-part-2}.

Even though there exists a computable bound on the size of the finite models produced by the construction, it was shown in \cite[Theorem~3.3]{Zakharyaschev1996-canonical-formulas-for-K4-part-2} that there are continuum-many subframe logics extending \(\logKTrans\), so in particular there are undecidable ones.
This is related to the partial order defined by \(\frmM \preccurlyeq \frmN\) \tdefiff{} there exists a definable sub-p-morphism from \(\frmM\) to \(\frmN\) \cite[Section~7.1]{Kracht1990-internal-definability-and-completeness-in-modal-logic}.
Similarly, we can define \(\frmM \preccurlyeq_{\lthSigma} \frmN\) \tdefiff{} there exists a definable \(\lthSigma\)-sub-p-morphism from \(\frmM\) to \(\frmN\). It is straightforward to show that starting with a model \(\frmM\) and applying our selection construction,  we obtain a \(\preccurlyeq_{\lthSigma}\)-minimal model. However, it is not clear whether it is the unique such model.

    \subsection*{Acknowledgements}

    We would like to thank the anonymous referees for their careful reading of our paper and helpful comments.

    \bibliography{references.bib}

\end{document}